\documentclass[12pt]{article}

\usepackage{amssymb, amsmath, amsthm, mathtools, commath}
\usepackage{mathrsfs, bbm, bm}

\usepackage{graphicx, subcaption, xcolor}

\usepackage{booktabs, array, float, multirow, arydshln}

\usepackage[T1]{fontenc}
\usepackage[utf8]{inputenc}
\usepackage{fullpage}
\usepackage{makeidx}
\usepackage{comment}
\usepackage{enumitem}

\usepackage[round]{natbib}
\usepackage[unicode,hidelinks]{hyperref}

\hypersetup{
  pdftitle={Toward design-based inference for data integration},
  pdfauthor={Andrius Čiginas, Ieva Burakauskaitė, Jae Kwang Kim},
  pdfkeywords={data integration; design-based inference; generalized regression estimator; non-probability sample; optimal sampling design; sequential sampling}
}

\theoremstyle{plain}
\newtheorem{proposition}{Proposition}
\newtheorem{theorem}{Theorem}

\newtheorem{lemma}{Lemma}

\theoremstyle{definition}
\newtheorem{remark}{Remark}

\newcommand{\bx}{\bm{x}}

\newcommand{\bbeta}{\bm{\beta}}

\newcommand{\bB}{\bm{B}}
\newcommand{\bOmega}{\bm{\Omega}}

\newcommand{\btau}{\bm{\tau}}
\newcommand{\bM}{\bm{M}}
\newcommand{\bA}{\bm{A}}

\newcommand{\bX}{\bm{X}}

\newenvironment{keywords}{%
  \par\noindent\textbf{Key words:}\ }{%
  \par\medskip}

\begin{document}

\title{\textbf{Toward design-based inference for data integration}}

\author{
Andrius~Čiginas\textsuperscript{1}\thanks{Corresponding author. E-mail: \texttt{andrius.ciginas@mif.vu.lt}.} \qquad
Ieva~Burakauskaitė\textsuperscript{1} \qquad
Jae~Kwang~Kim\textsuperscript{2}\\[0.4em]
{\small \textsuperscript{1}Faculty of Mathematics and Informatics, Vilnius University, Vilnius, Lithuania}\\[0.25em]
{\small \textsuperscript{2}Department of Statistics, Iowa State University, Ames, USA}
}

\date{}

\maketitle

\begin{abstract}

Integrating non-probability samples into finite-population inference typically requires modeling unknown selection probabilities under a missing-at-random (MAR) assumption that is difficult to verify. We propose a design-based alternative in which the non-probability sample is treated as a fully observed certainty stratum and a probability sample is drawn only from the complementary, previously unsampled units. Within this sequential framework, we develop two generalized regression estimators---one fitting the outcome model separately in the complementary stratum, the other pooling both samples---and make two distinct contributions. First, both estimators are design-consistent and admit consistent variance estimators with no assumption whatsoever on the non-probability selection mechanism, including under not-missing-at-random (NMAR) selection. Second, under a working superpopulation model that holds in both strata, the pilot non-probability sample can be used to construct second-stage inclusion probabilities that achieve Isaki--Fuller asymptotic optimality for the separate estimator; this optimality claim relies on assumptions strictly stronger than MAR, but its failure does not invalidate the consistency results above. A diagnostic test for coefficient homogeneity is proposed to guide the choice between the two estimators. Simulations confirm that the sequential estimators remain essentially unbiased under both MAR and NMAR, while propensity-adjusted competitors can be severely biased under NMAR. Two applications from Lithuanian official statistics illustrate that separate regression is preferable when the pilot stratum and its complement are strongly heterogeneous, whereas combined regression offers a modest efficiency gain when the two strata are similar.

\end{abstract}


\section{Introduction}

In recent years, statisticians have witnessed growing interest in combining different data sources to produce more accurate and cost-effective estimates of population parameters. Survey sampling, including official statistics surveys, traditionally relies on probability samples, where selection probabilities are known by design, ensuring the validity and unbiasedness of statistical inference. However, the increasing costs and declining response rates associated with probability sampling have stimulated the exploration of alternative data sources. Among these alternatives, non-probability samples such as voluntary samples and incomplete administrative data have gained popularity due to their relatively low cost and ease of data collection. Nevertheless, integrating non-probability samples into sample surveys presents substantial methodological challenges, primarily because the underlying selection mechanism is usually unknown, leading to significant selection biases if not properly addressed.

The prevailing statistical literature typically tackles these biases by invoking strong assumptions such as missing-at-random (MAR) and positivity (common support condition). Under these assumptions, adjustments through weighting or modeling can yield approximately unbiased and consistent estimators. However, such assumptions are often questionable and challenging to verify in practice, particularly when critical determinants of selection are unobserved or when the population of interest includes heterogeneous subgroups. Thus, the robustness of estimators relying on these assumptions remains uncertain, motivating the need for alternative strategies that rely less on unverifiable conditions.

We propose a sequential design-based framework in which the non-probability sample is treated as a fully observed certainty stratum, and the probability sample is then drawn only from the complementary, previously unsampled units. Within this framework, we make two distinct contributions, each supported by its own set of assumptions and addressing a different aspect of finite-population inference.

The first contribution concerns the \emph{validity} of inference. We establish that the proposed generalized regression (GREG) estimators are design-consistent under standard regularity conditions on the second-stage probability design and admit consistent plug-in variance estimators, with no assumption whatsoever on the selection mechanism of the non-probability sample. In particular, design consistency continues to hold under not-missing-at-random (NMAR) selection, where standard inverse probability weighting and doubly robust estimators can be severely biased. This robustness reflects a structural property of the proposed framework rather than a modeling choice: once the non-probability sample is treated as a certainty stratum, its inclusion probabilities never enter the estimator and therefore need not be modeled at all.

The second contribution concerns the \emph{efficiency} of the probability sample design. Under a working superpopulation model whose mean and variance functions hold in both strata, the non-probability sample serves as a pilot from which variance components can be estimated and second-stage inclusion probabilities can be constructed. In particular, for the GREG estimator that fits the outcome model separately in the complementary stratum (the separate GREG estimator), we show that the resulting design achieves Isaki--Fuller asymptotic optimality \citep{isaki1982}. This optimality claim relies on a model homogeneity assumption across strata that is strictly stronger than MAR. Importantly, however, the two contributions are layered rather than coupled: when the homogeneity assumption fails, the optimality guarantee is lost, but the design-consistency and variance-estimation results from the first contribution remain intact. The pilot-based design, therefore, offers efficiency gains without compromising inference robustness.

A related design-based data integration approach was previously considered by \citet{kimtam2020}, who similarly avoided modeling the unknown selection probabilities by post-stratifying on non-probability sample membership. In their framework, however, the probability sample is drawn from the entire population and may overlap with the non-probability sample, leading to duplicated observations, reduced efficiency, and the need for a fusion tuning parameter combining the two sources. Our sequential design draws the probability sample only from the complement of the non-probability sample, eliminating redundancy, removing the tuning parameter, and enabling the pilot to inform the second-stage design directly.

The practical relevance of our approach is demonstrated through real-world data applications. We focus on scenarios commonly encountered in official statistics, where auxiliary variables from administrative registers are known for all units in the target population. Such variables allow us to design the probability sampling step efficiently and to construct the regression-based estimators by explicitly accounting for population variability.

The rest of the paper proceeds as follows. Section~\ref{sec:2} introduces the measurement framework and reviews the independent-sampling literature, highlighting its reliance on unverifiable MAR assumptions and the problem of duplicate observations. Section~\ref{sec:3} presents the proposed sequential sampling framework and the resulting estimators. Section~\ref{sec:4} details the variance function estimation and the coefficient homogeneity test. Section~\ref{sec:5} establishes the theoretical properties of the proposed estimators, organized along the two contributions described above: design consistency without modeling the non-probability mechanism, and asymptotic optimality of the pilot-based design under a working superpopulation model. Sections~\ref{sec:6} and~\ref{sec:7} provide simulations based on artificial and real data, respectively, and Section~\ref{sec:8} concludes with a summary of findings.

\section{Basic setup and review of literature}\label{sec:2}

\subsection{Setup and notation}

Consider a finite population $U = \{1, 2, \dots, N\}$ of $N$ units, and let $y_i$ denote the value of the study variable $y$ for unit $i \in U$. Our target parameter is the finite population total
\begin{equation*}
Y = \sum_{i\in U} y_i.
\end{equation*}

For each unit $i \in U$, an auxiliary vector $\bx_i$ is available, typically from an administrative register. The study variable $y$, in contrast, is observed only through samples from $U$. We consider two such samples: a probability sample $S_p$ of size $n_p = |S_p|$, drawn under a known design with positive first-order inclusion probabilities, and a non-probability sample $S_{np}$ of size $n_{np} = |S_{np}|$, obtained through an unknown selection mechanism such as a voluntary response or an administrative source.

Two scenarios are commonly considered, depending on where $y$ is recorded. In the first, more commonly studied scenario, $y$ is observed only on $S_{np}$, while $S_p$ records only $\bx_i$ and the design weights. The probability sample then plays a purely referential role: it is used to model the unknown selection mechanism of $S_{np}$ as a function of $\bx$ \citep{EV_2017,kimwang2019}, and the resulting propensity estimates serve to debias the non-probability sample. This framework underpins the inverse probability weighting (IPW) and doubly robust (DR) estimators of \citet{CLW_2020} reviewed in Section~\ref{sec:independent}.

In the second scenario, which we adopt throughout the paper, $y$ is observed on both $S_p$ and $S_{np}$. The probability sample then contributes direct observations of $y$ in addition to its design-based information about the population, and the two sources can be combined to estimate $Y$ more efficiently than either alone. This motivates the fusion-type estimator reviewed in Section~\ref{sec:independent} and the sequential estimators developed in Section~\ref{sec:3}.

\subsection{Independent sampling framework and its limitations}\label{sec:independent}

A common setting in the literature is that $S_p$ and $S_{np}$ are obtained independently from $U$; see, e.g., \citet{kimtam2020}, \citet{GR_2024}, and \citet{CKN_2025}. Because the selection mechanism of $S_{np}$ is unknown, inference typically rests on the MAR assumption: the unknown inclusion probabilities $\pi_i^{(np)}$ depend only on observed auxiliaries $\bx$. Under MAR, one fits a propensity score model to obtain estimates $\hat{\pi}_i^{(np)}$, $i \in S_{np}$, and constructs the IPW estimator
\begin{equation*}
  \hat{Y}_{\text{IPW}} = \sum_{i\in S_{np}}\frac{1}{\hat{\pi}_i^{(np)}}y_i.
\end{equation*}
Using an outcome regression model fitted on $S_{np}$, the IPW estimator can be augmented to obtain the DR estimator
\begin{equation*}
  \hat{Y}_{\text{DR}} = \hat{Y}_{\text{IPW}} + \left( \sum_{i\in U} \bx_i - \sum_{i \in S_{np}} \frac{1}{\hat{\pi}_i^{(np)}} \bx_i \right)^\top \hat{\bbeta}_{np},
\end{equation*} 
where $\hat{\bbeta}_{np}$ denotes the corresponding regression coefficient estimator. The DR estimator is consistent if either the outcome regression or the propensity model is correctly specified \citep{CLW_2020}. 

When $y$ is observed in both samples, the DR estimator can be naturally combined with the probability-sample GREG estimator
\begin{equation*}
  \hat{Y}_{\text{GREG}} =
  \sum_{i\in S_p}\frac{1}{\pi_i^{(p)}}y_i
  + \left( \sum_{i\in U} \bx_i 
  - \sum_{i\in S_p}\frac{1}{\pi_i^{(p)}}\bx_i \right)^\top \hat{\bbeta}_{p},
\end{equation*}
where $\pi_i^{(p)}$ are the first-order inclusion probabilities of $S_p$ and $\hat{\bbeta}_{p}$ is the regression coefficient estimator fitted on $S_p$. This leads to the fusion estimator
\begin{equation}\label{eq:fusion}
  \hat{Y}_{\text{GREG-DR},\alpha}
    = \alpha\,\hat{Y}_{\text{GREG}} + (1-\alpha)\,\hat{Y}_{\text{DR}},
    \quad \alpha\in[0,1].
\end{equation}

This framework has three interrelated limitations that motivate our proposal. First, because $S_p$ is drawn from all of $U$, overlap with $S_{np}$ is unavoidable, and duplicated observations reduce estimation efficiency, particularly when $S_{np}$ is large. Second, the MAR assumption is inherently unverifiable and can lead to severe bias when the selection mechanism depends on $y$ itself, that is, under NMAR selection. Third, the fusion weight $\alpha$ in~\eqref{eq:fusion} is rarely chosen optimally in practice. The variance-minimizing $\alpha$ depends on the variances and covariance of $\hat{Y}_{\text{GREG}}$ and $\hat{Y}_{\text{DR}}$, which must themselves be estimated; simple substitutes such as $\alpha = n_p/(n_p+n_{np})$ are commonly used but are rarely variance-minimizing, and even the variance-minimizing choice does not address the first two limitations.

\section{Proposed method}\label{sec:3}

We propose a sequential sampling framework: the probability sample is drawn after, and from the complement of, the non-probability sample.

\begin{description}
  \item[Step 1.] A non-probability sample $S_{np} \subset U$ of size $n_{np}$ is obtained through an unknown selection mechanism, and the study variable $y$ is measured for all $i \in S_{np}$. Denote the remaining units by $U_1 = U \setminus S_{np}$, with $N_1 = |U_1|$.

  \item[Step 2.] A probability sample $S_p \subset U_1$ of size $n_p$ is then drawn under a known design with first-order inclusion probabilities $\pi_i^{(p)}$, and $y$ is measured for all $i \in S_p$.
\end{description}

This partition of $U$ into the certainty stratum $S_{np}$ and the complementary stratum $U_1$ eliminates three difficulties of the independent-sampling framework simultaneously: it removes the overlap between the two data sources, removes the need to model $\pi_i^{(np)}$, and removes the fusion weight $\alpha$ in \eqref{eq:fusion}. As we show in Section~\ref{sec:5}, the resulting estimators are design-consistent without any assumption on the selection mechanism of $S_{np}$, including under NMAR selection.

A closely related design-based approach is that of \citet{kimtam2020}, who also post-stratify on non-probability sample membership but draw $S_p$ from the full population $U$, allowing overlap with $S_{np}$. Our sequential design eliminates this redundancy and additionally enables the pilot sample to inform the second-stage design (see Section~\ref{sec:4}).

Since $S_{np}$ is treated as a certainty stratum, standard stratified sampling theory applies directly. In the absence of auxiliary information, a design-consistent estimator of $Y$ is
\begin{equation}
  \hat{Y}_{\text{DI}} = \sum_{i \in S_{np}} y_i + N_1
  \frac{\sum_{i \in S_p} y_i / \pi_i^{(p)}}{\sum_{i \in S_p} 1/\pi_i^{(p)}},
  \label{eq:DI}
\end{equation}
where $S_p \subset U_1$ is any probability sample. The estimator \eqref{eq:DI} is a stratified estimator: $S_{np}$ contributes the exact total $\sum_{i \in S_{np}} y_i$ as a self-representing certainty stratum, while the second term is a H\'ajek estimator of $\sum_{i \in U_1} y_i$ based on $S_p$.

When auxiliary information $\bx_i$ is available for all $i \in U$, incorporating it through a stratified regression estimator delivers substantial efficiency gains:
\begin{equation}\label{sepDI}
  \hat{Y}_{\text{sepDI}} = \sum_{i \in S_{np}} y_i + \sum_{i \in S_p} \frac{1}{\pi_i^{(p)}} y_i + \left( \sum_{i \in U_1} \bx_i - \sum_{i \in S_p} \frac{1}{\pi_i^{(p)}} \bx_i \right)^\top \widehat{\bB}_q,
\end{equation} 
where 
\begin{equation}\label{eq:est_B_q}
\widehat{\bB}_q = \left( \sum_{i \in S_p} q_i \bx_i \bx_i^{\top} \right)^{-1} \sum_{i \in S_p} q_i \bx_i y_i
\end{equation}
estimates the stratum-$U_1$ regression coefficient under weights $q_i = q(\bx_i, \pi_i^{(p)})$, $i\in U_1$, to be specified. The corresponding population-level quantity
\begin{equation}\label{eq:B_q}
\bB_q = \left( \sum_{i \in U_1} \pi_i^{(p)} q_i \bx_i \bx_i^{\top} \right)^{-1} \sum_{i \in U_1} \pi_i^{(p)} q_i \bx_i y_i
\end{equation}
is the probability limit of $\widehat{\bB}_q$ as $n_p$ grows. Note that $\bB_q$ is a design quantity defined on the finite population $U_1$, not the superpopulation regression coefficient $\bbeta$ in \eqref{eq:superpop_model}; the two coincide only as probability limits under a correctly specified model. The estimator \eqref{sepDI} is approximately design-unbiased for a wide range of choices of $q_i$, including $q_i = 1/\pi_i^{(p)}$.

The regression coefficient can also be estimated from the combined sample $S = S_{np} \cup S_p$, yielding
\begin{equation}\label{comDI}
  \hat{Y}_{\text{comDI}} = \sum_{i \in S_{np}} y_i + \sum_{i \in S_p} \frac{1}{\pi_i^{(p)}} y_i + \left( \sum_{i \in U_1} \bx_i - \sum_{i \in S_p} \frac{1}{\pi_i^{(p)}} \bx_i \right)^\top \widehat{\bB}_{q; \text{com}},
\end{equation} 
where 
\begin{equation}\label{eq:est_B_q_com}
\widehat{\bB}_{q; \text{com}} = \left( \sum_{i \in S} q_i \bx_i \bx_i^{\top} \right)^{-1} \sum_{i \in S} q_i \bx_i y_i,
\end{equation}
based on weights $q_i = q(\bx_i, \pi_i)$ with
\begin{equation*}
    \pi_i = \begin{cases}
    1 &\text{if $i \in S_{np}$,} \\
    \pi_i^{(p)} &\text{if $i \in U_1$.} 
    \end{cases}
\end{equation*} 
The convention $\pi_i = 1$ on $S_{np}$ reflects certainty-stratum membership, not a probability statement about the unknown selection mechanism: it merely encodes the fact that every unit of $S_{np}$ enters the estimator with weight one. Equivalently, $q_i = q(\bx_i, 1)$ on $S_{np}$ and $q_i = q(\bx_i, \pi_i^{(p)})$ on $U_1$, so $q_i$ is piecewise-defined across the two strata; under the certainty-stratum convention, however, the two pieces enter \eqref{eq:est_B_q_com} on equal footing as a single weight sequence indexed by $i \in S$.

Equation \eqref{comDI} is a classical combined GREG estimator, as can be seen by writing
\begin{equation*}
  \hat{Y}_{\text{comDI}} = \sum_{i \in S} \frac{1}{\pi_i} y_i + \left( \sum_{i\in U} \bx_i - \sum_{i \in S} \frac{1}{\pi_i} \bx_i \right)^\top \widehat{\bB}_{q; \text{com}}.
\end{equation*}
The combined estimator \eqref{comDI} can be more efficient than the separate estimator \eqref{sepDI} when $S_{np}$ is large and the regression coefficients of \eqref{eq:superpop_model} are the same in both strata. This coefficient homogeneity is implied by MAR together with correct specification of the linear mean function in \eqref{eq:superpop_model}, but the two conditions are not equivalent: MAR alone need not yield a common $\bbeta$ when the linear mean is misspecified, and a common $\bbeta$ can hold under NMAR provided the model is correctly specified in both strata. Crucially, however, the combined GREG estimator \eqref{comDI} remains design-consistent regardless of whether the homogeneity assumption holds and regardless of the selection mechanism of $S_{np}$ (see Section~\ref{sec:5}). Homogeneity becomes relevant for efficiency, not for consistency.

We now turn to the second-stage design itself. Following \citet{isaki1982}, we construct inclusion probabilities $\pi_i^{(p)}$ that minimize anticipated variance under a working superpopulation model $\xi$. We assume the values $y_i$, $i \in U$, are independent draws from $\xi$ with conditional moments
\begin{equation}\label{eq:superpop_model} 
E_\xi(y_i \mid \bx_i) = \bx_i^\top \bbeta, \qquad V_\xi(y_i \mid \bx_i) = \sigma^2 v_i, 
\end{equation} 
where $v_i = v(\bx_i)$ is a known function of auxiliary variables that may involve unknown parameters. Under \eqref{eq:superpop_model}, \citet{isaki1982} show that the optimal design sets inclusion probabilities proportional to the conditional standard deviation:
\begin{equation}\label{eq:opt_pi} 
\pi_i^{(p)} \propto \{V_\xi(y_i \mid \bx_i)\}^{1/2} = \sigma v_i^{1/2}, \quad i \in U_1. 
\end{equation}

The optimal probabilities \eqref{eq:opt_pi} are not directly available, since the variance function involves unknown parameters. We use the pilot non-probability sample $S_{np}$ to estimate them; standard methods from the heteroscedastic-regression literature \citep[e.g.,][]{DC_1987} can be adapted by treating residuals from a pilot fit as responses in a variance regression. The details, including a specific working choice of $v_i$, are given in Section~\ref{sec:4}.

This pilot-based strategy relies on homogeneity of the working superpopulation model across strata: the parameters entering the mean and variance functions in \eqref{eq:superpop_model} are the same on $S_{np}$ as on $U_1$. This is stronger than the coefficient homogeneity discussed above. Under MAR, together with the correct specification of the full conditional model, such homogeneity holds, but MAR alone does not imply it, and it may also hold under NMAR when the working model is correctly specified in both strata. Even if this homogeneity fails, the proposed estimators remain design-consistent, and our inference remains design-based; what is lost is the optimality guarantee, not the validity of the inference.

\section{Estimation details}\label{sec:4}

This section specifies a working choice of variance model, shows how to estimate it from the pilot sample, and develops a diagnostic test for coefficient homogeneity to inform the choice between $\hat{Y}_{\text{sepDI}}$ and $\hat{Y}_{\text{comDI}}$.

For expositional simplicity, we adopt the power variance model
\begin{equation}\label{eq:pow_mod}
V_\xi(y_i \mid \bx_i) = \sigma^2 (\bx_i^\top \bbeta)^\gamma, \qquad \gamma > 0,
\end{equation}
which reflects the common pattern that variance increases with the mean. This specification is standard in heteroscedastic regression \citep[e.g.,][]{DC_1987} and matches patterns observed in our real-data application in Section~\ref{sec:7}. The model \eqref{eq:pow_mod} requires the linear predictor $\bx_i^\top \bbeta$ to be positive on $U_1$, a condition formalized as (L2) in Section~\ref{sec:5.3}.

\paragraph{Variance function estimation.} The variance function is estimated from $S_{np}$ in two stages. In the first stage, the ordinary least squares (OLS) estimator
\begin{equation}\label{eq:pilot_beta}
\hat{\bbeta}_{np} = \left( \sum_{i \in S_{np}} \bx_i \bx_i^\top \right)^{-1} \sum_{i \in S_{np}} \bx_i y_i
\end{equation}
yields residuals $\hat{e}_i = y_i - \bx_i^\top \hat{\bbeta}_{np}$ for $i \in S_{np}$. This estimator is consistent under mild conditions, even when heteroscedasticity is present. In the second stage, we model the squared residuals through the log-log regression
\begin{equation*}
\log(\hat{e}_i^2) = \log(\sigma^2) + \gamma \log(\hat{m}_i) + \varepsilon_i^*, \qquad \hat{m}_i = \bx_i^\top \hat{\bbeta}_{np},
\end{equation*}
with error term $\varepsilon_i^*$, and obtain $(\hat\sigma^2, \hat\gamma)$ by OLS. The estimated variance function for any unit $i \in U$ is then
\begin{equation}\label{eq:pilot_var}
\hat{\sigma}_i^2 = \hat{V}_\xi(y_i \mid \bx_i) = \hat{\sigma}^2 (\bx_i^\top \hat{\bbeta}_{np})^{\hat\gamma}.
\end{equation}
For more efficient estimates of both the regression coefficient and the error variances, $\hat{\bbeta}_{np}$ can be refined by feasible generalized least squares (FGLS):
\begin{equation}\label{eq:beta_FGLS}
\hat{\bbeta}_{np}^{\mathrm{FGLS}} = \left( \sum_{i \in S_{np}} \frac{\bx_i \bx_i^\top}{\hat{\sigma}_i^2} \right)^{-1} \sum_{i \in S_{np}} \frac{\bx_i y_i}{\hat{\sigma}_i^2}.
\end{equation}
The variance estimates $\hat{\sigma}_i^2$ are then re-computed using \eqref{eq:beta_FGLS}, and the procedure may be iterated to convergence if desired.

\paragraph{Optimal inclusion probabilities.} For the separate regression estimator $\hat{Y}_{\text{sepDI}}$ in \eqref{sepDI}, the design with estimated inclusion probabilities
\begin{equation}\label{eq:est_opt_pi0}
\hat{\pi}_i^{(p)} = n_p \frac{\hat{\sigma}_i}{\sum_{j \in U_1} \hat{\sigma}_j}, \qquad i \in U_1,
\end{equation}
where $\hat{\sigma}_i = (\hat{\sigma}_i^2)^{1/2}$ with $\hat{\sigma}_i^2$ from \eqref{eq:pilot_var}, achieves Isaki--Fuller asymptotic optimality for $\hat{Y}_{\text{sepDI}}$ under the assumptions formalized in Section~\ref{sec:5.3} (Theorem~\ref{th:1}).

\paragraph{Deciding between separate and combined estimators.} The combined estimator $\hat{Y}_{\text{comDI}}$ pools $S_{np}$ and $S_p$ in estimating the regression coefficient, which may yield efficiency gains under coefficient homogeneity but does not guarantee improvement otherwise. To inform the choice between $\hat{Y}_{\text{sepDI}}$ and $\hat{Y}_{\text{comDI}}$, we test
\begin{equation*}
H_0:\ \bbeta_{np} = \bbeta_p \qquad \text{against} \qquad H_1:\ \bbeta_{np} \neq \bbeta_p,
\end{equation*}
where $\bbeta_{np}$ and $\bbeta_p$ denote the regression coefficients on $S_{np}$ and $U_1$ respectively. The test statistic compares the FGLS estimators on the two strata.

For the FGLS estimator $\hat{\bbeta}_{np}^{\mathrm{FGLS}}$ in \eqref{eq:beta_FGLS}, the variance is estimated by the sandwich formula
\begin{equation}\label{eq:Vbeta_np}
\hat{V}(\hat{\bbeta}_{np}^{\mathrm{FGLS}}) = \left( \sum_{i \in S_{np}} \frac{\bx_i \bx_i^\top}{\hat{\sigma}_i^2} \right)^{-1} \left( \sum_{i \in S_{np}} \frac{\bx_i \bx_i^\top}{\hat{\sigma}_i^4} \hat{e}_i^2 \right) \left( \sum_{i \in S_{np}} \frac{\bx_i \bx_i^\top}{\hat{\sigma}_i^2} \right)^{-1},
\end{equation}
with $\hat{e}_i = y_i - \bx_i^\top \hat{\bbeta}_{np}^{\mathrm{FGLS}}$. This estimator is consistent under mild regularity conditions, even if the variance model is misspecified.

For stratum $U_1$, the regression coefficient $\bbeta_p$ is estimated from $S_p$ by an analogous procedure, weighted additionally by the inverse inclusion probabilities $\pi_i^{(p)}$ (e.g., $\hat\pi_i^{(p)}$ from \eqref{eq:est_opt_pi0}). Starting from
\begin{equation*}
\hat{\bbeta}_p = \left( \sum_{i \in S_p} \frac{\bx_i \bx_i^\top}{\pi_i^{(p)}} \right)^{-1} \sum_{i \in S_p} \frac{\bx_i y_i}{\pi_i^{(p)}},
\end{equation*}
the iterative FGLS procedure yields
\begin{equation}\label{eq:beta_p_FGLS}
\hat{\bbeta}_p^{\mathrm{FGLS}} = \left( \sum_{i \in S_p} \frac{\bx_i \bx_i^\top}{\pi_i^{(p)} \hat{\tau}_i^2} \right)^{-1} \sum_{i \in S_p} \frac{\bx_i y_i}{\pi_i^{(p)} \hat{\tau}_i^2},
\end{equation}
where $\hat{\tau}_i^2$ are predicted error variances from a power variance model fitted on $S_p$. Thus, for the homogeneity test, we work with the choice $q_i = 1/(\pi_i^{(p)} \hat{\tau}_i^2)$; Section~\ref{sec:5} continues to allow a general $q_i$.

The variance of $\hat{\bbeta}_p^{\mathrm{FGLS}}$ has both design and model components. Letting $\mathcal{F}_1 = \{(\bx_i, y_i) : i \in U_1\}$, the standard decomposition gives
\begin{equation*}
V(\hat{\bbeta}_p^{\mathrm{FGLS}}) = E_\xi \{V(\hat{\bbeta}_p^{\mathrm{FGLS}} \mid \mathcal{F}_1)\} + V_\xi \{E(\hat{\bbeta}_p^{\mathrm{FGLS}} \mid \mathcal{F}_1)\},
\end{equation*}
where the first term reflects sampling variability and the second captures model variability. Estimating each component separately,
\begin{equation}\label{eq:Vbeta_p}
\hat{V}(\hat{\bbeta}_p^{\mathrm{FGLS}}) = \hat{V}(\hat{\bbeta}_p^{\mathrm{FGLS}} \mid \mathcal{F}_1) + \hat\bM_\tau^{-1} \left( \sum_{i \in S_p} \frac{\bx_i \bx_i^\top}{\pi_i^{(p)} \hat{\tau}_i^4} \hat{e}_i^2 \right) \hat\bM_\tau^{-1},
\end{equation}
where $\hat{e}_i = y_i - \bx_i^\top \hat{\bbeta}_p^{\mathrm{FGLS}}$ and $\hat\bM_\tau = \sum_{i \in S_p} \bx_i \bx_i^\top / (\pi_i^{(p)} \hat{\tau}_i^2)$. The first term is the design-based variance estimator
\begin{equation*}
\hat{V}(\hat{\bbeta}_p^{\mathrm{FGLS}} \mid \mathcal{F}_1) = \hat\bM_\tau^{-1} \left[ \sum_{i \in S_p} \sum_{j \in S_p} \frac{\Delta_{ij}^{(p)}}{\pi_{ij}^{(p)}} \frac{\bx_i \hat{e}_i}{\pi_i^{(p)} \hat{\tau}_i^2} \frac{\bx_j^\top \hat{e}_j}{\pi_j^{(p)} \hat{\tau}_j^2} \right] \hat\bM_\tau^{-1},
\end{equation*}
with $\Delta_{ij}^{(p)} = \pi_{ij}^{(p)} - \pi_i^{(p)} \pi_j^{(p)}$ and $\pi_{ij}^{(p)} > 0$ the second-order inclusion probabilities. The second term in \eqref{eq:Vbeta_p} is a model-variance sandwich that is asymptotically negligible under standard conditions.

Based on \eqref{eq:beta_FGLS}, \eqref{eq:Vbeta_np}, \eqref{eq:beta_p_FGLS}, and \eqref{eq:Vbeta_p}, the test statistic is
\begin{equation}\label{eq:Ftest}
F = \left( \hat{\bbeta}_{np}^{\mathrm{FGLS}} - \hat{\bbeta}_p^{\mathrm{FGLS}} \right)^\top \left\{ \hat{V}(\hat{\bbeta}_{np}^{\mathrm{FGLS}}) + \hat{V}(\hat{\bbeta}_p^{\mathrm{FGLS}}) \right\}^{-1} \left( \hat{\bbeta}_{np}^{\mathrm{FGLS}} - \hat{\bbeta}_p^{\mathrm{FGLS}} \right).
\end{equation}
The additive form of the variance in \eqref{eq:Ftest} relies on the (asymptotic) independence of $\hat{\bbeta}_{np}^{\mathrm{FGLS}}$ and $\hat{\bbeta}_p^{\mathrm{FGLS}}$, which holds in the sequential framework: conditional on $S_{np}$, the second-stage design draws $S_p$ from $U_1 = U \setminus S_{np}$, so the two estimators are computed from disjoint subsets of the population and the design randomness in $\hat{\bbeta}_p^{\mathrm{FGLS}}$ is independent of $\hat{\bbeta}_{np}^{\mathrm{FGLS}}$.

For implementation, since the model-variance contribution to \eqref{eq:Vbeta_p} is asymptotically negligible, we use only the design-based component $\hat{V}(\hat{\bbeta}_p^{\mathrm{FGLS}} \mid \mathcal{F}_1)$ in computing \eqref{eq:Ftest}; this matches the implementation used in the simulations and applications of Sections~\ref{sec:6}--\ref{sec:7}. Under $H_0$, the resulting statistic \eqref{eq:Ftest} is asymptotically chi-square with degrees of freedom equal to $\dim(\bbeta)$. Large values of $F$ provide evidence against coefficient homogeneity and favor $\hat{Y}_{\text{sepDI}}$; if $H_0$ is not rejected, the combined estimator $\hat{Y}_{\text{comDI}}$ may be preferred for its efficiency gain.

\section{Statistical properties}\label{sec:5}

In this section, we examine the theoretical properties of the proposed estimators, organized along the two contributions described in the introduction. Section~\ref{sec:theory} establishes design consistency and consistent variance estimation for the separate and combined estimators, with no assumption on the selection mechanism of the non-probability sample. Section~\ref{sec:5.3} then establishes Isaki--Fuller asymptotic optimality of the pilot-based second-stage design for the separate estimator under a working superpopulation model that holds in both strata. The model assumptions required for the optimality result are strictly stronger than MAR, but their failure affects only the optimality guarantee, not the design-consistency and variance-estimation results of Section~\ref{sec:theory}.

Both subsections exploit the same structural feature of the sequential framework: once the non-probability sample is treated as a certainty stratum, the separate and combined estimators are classical stratified GREG estimators with known inclusion probabilities, so standard design-based arguments apply directly. The Isaki--Fuller optimality strategy, however, applies only to the separate estimator. The combined estimator includes model components derived from the fixed non-probability stratum, whose inclusion probabilities cannot be controlled or optimized; while it can be more efficient under regression coefficient homogeneity, it is not amenable to direct optimality derivations due to its fixed design structure.

\subsection{Asymptotic properties}\label{sec:theory}

This subsection establishes our first theoretical contribution: design consistency, asymptotic variance formulas, and consistent variance estimation for both the separate and combined estimators, without any assumption on the selection mechanism of the non-probability sample. The arguments are entirely design-based and exploit the certainty-stratum interpretation of $S_{np}$.

Both estimators \eqref{sepDI} and \eqref{comDI} take the form of a classical stratified GREG estimator, differing only in whether the regression coefficient for stratum $U_1$ is estimated from $S_p$ alone ($\widehat{\bB}_q$) or from the pooled sample $S$ ($\widehat{\bB}_{q;\text{com}}$). For the separate estimator, the relevant population quantity is $\bB_q$ defined in \eqref{eq:B_q}. For the combined estimator, the subsequent linearization is stated around a pseudo-target $\bB_q^\star$ satisfying $\widehat{\bB}_{q;\text{com}} - \bB_q^\star = o_p(1)$. In the homogeneous-model case covered by Proposition~\ref{prop:consistency-combined} in Appendix~\ref{app:proofs}, one may take $\bB_q^\star = \bB_q$.

For the complementary stratum, let $Y_{U_1} = \sum_{i\in U_1} y_i$ and $\bX_{U_1} = \sum_{i\in U_1} \bx_i$, and let $\hat{Y}_{\mathrm{HT},U_1} = \sum_{i\in S_p} y_i / \pi_i^{(p)}$ and $\hat{\bX}_{\mathrm{HT},U_1} = \sum_{i\in S_p} \bx_i / \pi_i^{(p)}$ denote the corresponding Horvitz--Thompson (HT) estimators.

\paragraph{Design consistency.} From \eqref{sepDI} and \eqref{comDI},
\begin{align*}
\hat{Y}_{\text{sepDI}} - Y &= \hat{Y}_{\mathrm{HT},U_1} - Y_{U_1} + (\bX_{U_1} - \hat{\bX}_{\mathrm{HT},U_1})^\top \widehat{\bB}_q,\\
\hat{Y}_{\text{comDI}} - Y &= \hat{Y}_{\mathrm{HT},U_1} - Y_{U_1} + (\bX_{U_1} - \hat{\bX}_{\mathrm{HT},U_1})^\top \widehat{\bB}_{q;\text{com}}.
\end{align*}
Hence, if $N_1^{-1} (\hat{Y}_{\mathrm{HT},U_1} - Y_{U_1}) = o_p(1)$ and $N_1^{-1} \|\hat{\bX}_{\mathrm{HT},U_1} - \bX_{U_1}\| = o_p(1)$, and the relevant regression coefficient estimator is stochastically bounded, then the estimator is design-consistent. For $\widehat{\bB}_q$, such boundedness is a classical consequence of standard GREG regularity \citep{Fuller2009}. For the combined estimator, the identity for $\hat{Y}_{\text{comDI}} - Y$ shows that stochastic boundedness of $\widehat{\bB}_{q;\text{com}}$ is sufficient for design consistency.

\paragraph{Linearization.} Together with $\|\widehat{\bB}_q - \bB_q\| = o_p(1)$, the exact expansion
\begin{equation}\label{eq:expand}
\hat{Y}_{\text{sepDI}} - Y = \hat{Y}_{\mathrm{HT},U_1} - Y_{U_1} + (\bX_{U_1} - \hat{\bX}_{\mathrm{HT},U_1})^\top \bB_q + (\bX_{U_1} - \hat{\bX}_{\mathrm{HT},U_1})^\top (\widehat{\bB}_q - \bB_q)
\end{equation}
has remainder term $o_p(n_p^{-1/2} N_1)$. The remainder rate uses
\begin{equation}\label{eq:rate}
N_1^{-1} \|\hat{\bX}_{\mathrm{HT},U_1} - \bX_{U_1}\| = O_p(n_p^{-1/2}),
\end{equation}
which holds for the second-stage design under standard regularity; in particular, \eqref{eq:rate} is implied by conditions (T1)--(T2) of Theorem~\ref{th:1} below for Poisson sampling on $U_1$ with bounded inclusion probabilities. If, similarly, $\|\widehat{\bB}_{q;\text{com}} - \bB_q^\star\| = o_p(1)$, the same conclusion holds for the combined estimator with $\bB_q^\star$ in place of $\bB_q$. The asymptotic variance formulas below therefore depend only on the leading terms; see \citet{Kim2025} for general background.

\paragraph{Asymptotic variances and their estimators.} Retaining only the leading terms in \eqref{eq:expand} and in the analogous expansion for the combined estimator, the design-based asymptotic variance of the separate estimator is
\begin{equation}\label{eq:varsep}
V(\hat{Y}_{\text{sepDI}}) = \sum_{i\in U_1} \sum_{j\in U_1} \Delta_{ij}^{(p)} \frac{y_i - \bx_i^\top \bB_q}{\pi_i^{(p)}} \frac{y_j - \bx_j^\top \bB_q}{\pi_j^{(p)}},
\end{equation}
and that of the combined estimator is
\begin{equation}\label{eq:varcom}
V(\hat{Y}_{\text{comDI}}) = \sum_{i\in U_1} \sum_{j\in U_1} \Delta_{ij}^{(p)} \frac{y_i - \bx_i^\top \bB_q^\star}{\pi_i^{(p)}} \frac{y_j - \bx_j^\top \bB_q^\star}{\pi_j^{(p)}},
\end{equation}
where $\Delta_{ij}^{(p)} = \pi_{ij}^{(p)} - \pi_i^{(p)} \pi_j^{(p)}$. When Proposition~\ref{prop:consistency-combined} holds, one may take $\bB_q^\star = \bB_q$, and \eqref{eq:varcom} reduces to \eqref{eq:varsep}; otherwise \eqref{eq:varcom} reflects the additional variance arising from a potentially misspecified pooled fit. Design-consistent plug-in estimators of \eqref{eq:varsep} and \eqref{eq:varcom} are obtained by replacing $\bB_q$ and $\bB_q^\star$ with $\widehat{\bB}_q$ and $\widehat{\bB}_{q;\text{com}}$, respectively, yielding
\begin{equation}\label{eq:varest}
\hat{V}(\hat{Y}_\bullet) = \sum_{i\in S_p} \sum_{j\in S_p} \frac{\Delta_{ij}^{(p)}}{\pi_{ij}^{(p)}} \frac{y_i - \bx_i^\top \widehat{\bB}_\bullet}{\pi_i^{(p)}} \frac{y_j - \bx_j^\top \widehat{\bB}_\bullet}{\pi_j^{(p)}},
\end{equation}
where $\widehat{\bB}_\bullet$ is $\widehat{\bB}_q$ or $\widehat{\bB}_{q;\text{com}}$ according to the estimator; consistency of \eqref{eq:varest} follows from standard arguments \citep{BO_2017}.

\subsection{Asymptotic optimality of the probability sample design}\label{sec:5.3}

This subsection establishes our second theoretical contribution: under a working superpopulation model that holds in both strata, the pilot-based second-stage design achieves Isaki--Fuller asymptotic optimality for the separate estimator. As emphasized at the start of Section~\ref{sec:5}, the required model assumptions are strictly stronger than MAR, but their failure affects only the optimality guarantee; the design-consistency and variance-estimation results of Section~\ref{sec:theory} remain intact.

We adopt the power variance model \eqref{eq:pow_mod} for concreteness. Under the assumptions of Lemma~\ref{lem:unifV}, the pilot-based plug-in variance \eqref{eq:pilot_var} is uniformly consistent on $U_1$, and Theorem~\ref{th:1} shows that the design using the estimated inclusion probabilities is asymptotically optimal in the anticipated-variance sense. Analogous conclusions hold for alternative variance functions once the corresponding uniform consistency result replaces Lemma~\ref{lem:unifV}; see Remark~\ref{rem:2}.

\begin{lemma}[Uniform consistency of the power-variance plug-in]\label{lem:unifV}
Assume, along a sequence with $n_{np}\to\infty$:
\begin{itemize}
\item[\textnormal{(L1)}] Covariates and the pilot design matrix satisfy $\sup_{i\in U}\|\bx_i\|\le C<\infty$ and
\begin{equation*}
\frac{1}{n_{np}}\sum_{i\in S_{np}}\bx_i\bx_i^\top \xrightarrow{p} \bOmega_{np},
\end{equation*}
where $\bOmega_{np}$ is positive definite.
\item[\textnormal{(L2)}] The linear predictor is positive on $U_1$: there exists $c>0$ such that $\inf_{i\in U_1}\bx_i^\top\bbeta \ge c$.
\item[\textnormal{(L3)}] The superpopulation mean model in \eqref{eq:superpop_model} and the power variance form \eqref{eq:pow_mod} hold in both strata $S_{np}$ and $U_1$; and exogeneity holds in $S_{np}$: with $e_i=y_i-\bx_i^\top\bbeta$,
\begin{equation*}
E_\xi (e_i \mid \bx_i,\ i\in S_{np})=0 .
\end{equation*}
\item[\textnormal{(L4)}] The outcomes satisfy $\sup_{i\in U}E_\xi\!\left(|y_i|^4\right)<\infty$.
\end{itemize}
Let $\hat{\bbeta}_{np}$ be the pilot estimator in \eqref{eq:pilot_beta}, and let $(\hat\gamma,\hat\sigma^2)$ be obtained by the log-log fit in Section~\ref{sec:4}, yielding $\hat V_\xi(y_i\mid\bx_i)$ in \eqref{eq:pilot_var}. Then
\begin{equation*}
\sup_{i\in U_1}\big|\,\hat V_\xi(y_i\mid\bx_i)-V_\xi(y_i\mid\bx_i)\,\big|\;\xrightarrow{p}\;0 .
\end{equation*}
\end{lemma}

\noindent\textit{Proof.} The proof is given in Appendix~\ref{app:lem-unifV}.

By Lemma~\ref{lem:unifV}, uniform consistency holds on $U_1$, which justifies replacing the infeasible optimal inclusion probabilities \eqref{eq:opt_pi} by their estimates. We now state the asymptotic optimality of the estimated design.

\begin{theorem}[Asymptotic optimality under estimated design]\label{th:1}
Assume that conditions \textnormal{(L1)}--\textnormal{(L4)} of Lemma~\ref{lem:unifV} hold. Along a sequence with $n_{np}\to\infty$, suppose:
\begin{itemize}
\item[\textnormal{(T1)}] The design matrix on $U_1$ satisfies
\begin{equation*}
\frac{1}{N_1}\sum_{i\in U_1}\bx_i\bx_i^\top \;\to\; \bM,
\end{equation*}
where $\bM$ is positive definite.
\item[\textnormal{(T2)}] The size of the probability sample $S_p$ satisfies $n_p/N_1\to f\in(0,1)$, and the inclusion probabilities on $U_1$ are uniformly bounded away from $0$ and $1$.
\end{itemize}
Let the design on $U_1$ use the estimated probabilities
\begin{equation}\label{eq:est_opt_pi}
\hat{\pi}_i^{(p)}
= n_p\,
\frac{\{\hat V_\xi(y_i\mid \bx_i)\}^{1/2}}{\sum_{j\in U_1}\{\hat V_\xi(y_j\mid \bx_j)\}^{1/2}},
\quad i\in U_1,
\end{equation}
where $\hat V_\xi(y_i\mid \bx_i)$ is the plug-in variance defined in \eqref{eq:pilot_var}. Denote by $\pi_i^{(p)}$ the optimal probabilities proportional to $\{V_\xi(y_i\mid\bx_i)\}^{1/2}$ as in \eqref{eq:opt_pi}. Then the design based on \eqref{eq:est_opt_pi} achieves asymptotic optimality for the separate GREG estimator $\hat{Y}_{\text{sepDI}}$ in \eqref{sepDI}, in the sense of minimizing the anticipated variance as in \citet{isaki1982}.
\end{theorem}

\noindent\textit{Proof.} The proof is given in Appendix~\ref{app:thm-opt}.

\begin{remark}\label{rem:2}
The asymptotic optimality result in Theorem~\ref{th:1} is stated under the power variance function \eqref{eq:pow_mod}. For an alternative parametric variance model (e.g., an exponential form), replace Lemma~\ref{lem:unifV} with a result asserting uniform consistency on $U_1$ of the variance plug-in fitted on $S_{np}$, and then apply the same continuity argument used in Theorem~\ref{th:1}. The regularity conditions should be adapted to the chosen variance model, typically including moment bounds (cf.\ (L4)) and, where appropriate, a condition analogous to (L2), while the overall proof strategy remains unchanged.
\end{remark}

\section{Simulation study}\label{sec:6}

We conduct a simulation study to (i) assess the design-based validity of the proposed sequential estimators in terms of bias, variance estimation, and confidence interval coverage, and (ii) illustrate how methods relying on MAR-type propensity score adjustments can break down when the non-probability mechanism is NMAR. We consider two mechanisms for generating the realized non-probability sample $S_{np}$: a MAR mechanism that depends only on auxiliary variables, and an NMAR mechanism that depends on the study variable in addition to the auxiliaries.

\subsection{Simulation setup}\label{sec:6:setup}

We generate a fixed finite population $U$ of size $N=10\,000$ once and treat it as fixed across Monte Carlo repetitions. For each unit $i\in U$, generate independent covariates $x_{1i}, x_{2i} \sim \mathcal{U}(0,1)$ and define the mean function with an interaction term
\begin{equation*}
\mu_i=\beta_0+\beta_1 x_{1i}+\beta_2 x_{2i}+\beta_3 x_{1i}x_{2i}.
\end{equation*}
The outcome is generated from a scaled lognormal model
\begin{equation*}
y_i=\mu_i\exp(\varepsilon_i), \quad \varepsilon_i\sim \mathcal{N}(-\sigma^2/2,\sigma^2),
\end{equation*}
so that $E_\xi(y_i\mid \bx_i)=\mu_i$ and $V_\xi(y_i\mid \bx_i)=\mu_i^2\{\exp(\sigma^2)-1\}$. Throughout, the regression-type estimators use the working covariate vector $\bx_i=(1,x_{1i},x_{2i})^\top$, deliberately omitting the interaction term in $\mu_i$. We set $(\beta_0,\beta_1,\beta_2,\beta_3)=(10,15,10,20)$ and $\sigma=0.6$.

Let $\delta_i\in\{0,1\}$ indicate membership in the non-probability sample $S_{np}=\{i:\delta_i=1\}$, and denote the remaining units by $U_1=U\setminus S_{np}$ with $N_1=|U_1|$. Conditional on the generated population, we draw the $\delta_i$ independently with selection probabilities $p_i$ given by one of the following models:

\smallskip
\noindent\emph{(i) MAR mechanism.} $\mathrm{logit}(p_i)=\alpha_0+\alpha_1 x_{1i}+\alpha_2 x_{2i}$, with $(\alpha_1,\alpha_2)=(2,-2)$.

\smallskip
\noindent\emph{(ii) NMAR mechanism.} $\mathrm{logit}(p_i)=\alpha_0+\alpha_1 x_{1i}+\alpha_2 x_{2i}+\alpha_3 \log(1+y_i)$, with $(\alpha_1,\alpha_2,\alpha_3)=(2,-2,0.5)$.

\smallskip
\noindent In each case, the intercept $\alpha_0$ is calibrated so that $N^{-1}\sum_{i\in U}p_i = f_{np}$, where $f_{np}$ is a target non-probability sampling rate. We set $f_{np}=0.70$, reflecting a regime in which the non-probability source is much larger than the probability sample to be drawn from its complement; this is the regime where data integration is most consequential and where the efficiency benefit of treating $S_{np}$ as a certainty stratum is most pronounced. Smaller values of $f_{np}$ are not expected to alter the qualitative conclusions, since for any $f_{np}\in(0,1)$ the certainty-stratum interpretation continues to drive the design-consistency of the sequential estimators.

Conditional on the realized $S_{np}$, we draw a second-stage Poisson sample $S_p\subset U_1$ with first-order inclusion probabilities $\pi_i^{(p)}$ and expected size $n_p=\lfloor f_p N_1\rfloor$, with $f_p=0.40$. The inclusion probabilities on $U_1$ are constructed from the pilot variance model of Section~\ref{sec:4}: we fit the power variance model \eqref{eq:pow_mod} on $S_{np}$ to obtain $\hat{\sigma}_i^2$ for all units, then set $\pi_i^{(p)}\propto \hat{\sigma}_i$ on $U_1$ and rescale to satisfy $\sum_{i\in U_1}\pi_i^{(p)}=n_p$ (with truncation at $1$ when necessary). The pilot variance fit uses a single FGLS update of $\hat{\bbeta}_{np}$ following the OLS pilot estimator \eqref{eq:pilot_beta}, without further iteration.

For design comparisons (used only to illustrate the benefit of the estimated optimal allocation), we also consider an equal-probability Poisson design with $\pi_i^{(p)}\equiv n_p/N_1$ and a Poisson probability proportional to size (PPS) design with $\pi_i^{(p)}\propto x_{1i}$.

For comparison with estimators developed under the independent-sampling framework of Section~\ref{sec:2}, we additionally draw an independent Poisson probability sample $S_p^{\mathrm{ind}}\subset U$ with equal first-order inclusion probabilities $\pi_i^{(\mathrm{ind})}\equiv n_p^{\mathrm{ind}}/N$ and expected size $n_p^{\mathrm{ind}}=\lfloor f_p(1-f_{np})N\rfloor$. This independent sample is generated without excluding the realized $S_{np}$ and may therefore overlap with it.

\subsection{Estimators compared}\label{sec:6:estimators}

We compare two groups of estimators of the population total $Y=\sum_{i\in U} y_i$.

\paragraph{Sequential (design-based) estimators.}
Using $(S_{np},S_p)$, we compute the data integration (DI) estimator $\hat{Y}_{\text{DI}}$ in \eqref{eq:DI}, the stratified HT estimator
\begin{equation}\label{eq:HT_seq}
\hat{Y}_{\text{HT}}^{\text{seq}}=\sum_{i\in S_{np}}y_i+\sum_{i\in S_p}\frac{1}{\pi_i^{(p)}}y_i,
\end{equation}
and regression-type estimators in \eqref{sepDI} and \eqref{comDI} under two specifications of the working weight,
\begin{equation*}
q_i^{(b)}=\frac{1}{\pi_i^{(p)}} \quad \text{and} \quad q_i^{(\sigma)}=\frac{1}{\pi_i^{(p)}\hat{\sigma}_i^2}.
\end{equation*}
We report $\hat{Y}_{\text{sepDI}}(q^{(b)})$, $\hat{Y}_{\text{sepDI}}(q^{(\sigma)})$, and $\hat{Y}_{\text{comDI}}(q^{(\sigma)})$. To guide the choice between separate and combined regression, we use the coefficient homogeneity test of Section~\ref{sec:4}, with $\mathcal{R}$ denoting the rejection event. The adaptive estimator is
\begin{equation}\label{eq:addi_rule}
\hat{Y}_{\text{adDI}}(q^{(\sigma)}) =
\begin{cases}
\hat{Y}_{\text{sepDI}}(q^{(\sigma)}) & \text{if }\mathcal{R}\text{ occurs,}\\
\hat{Y}_{\text{comDI}}(q^{(\sigma)}) & \text{otherwise.}
\end{cases}
\end{equation}

\paragraph{Numerical safeguards.} The pilot variance fit and the resulting weights $q_i^{(\sigma)}$ are sensitive to extreme values of the fitted variance, so we apply the following safeguards in all sequential estimators: (i) cap $|\hat{\gamma}|\le 3$; (ii) floor pilot mean predictions at the empirical $5\%$ quantile of positive values in $S_{np}$; (iii) bound $\hat{\sigma}_i^2$ away from zero; (iv) enforce $\pi_i^{(p)}\ge 0.01$ in the estimated optimal Poisson design; and (v) truncate the working weights $q_i^{(\sigma)}$ at the empirical $99.9$th percentile within the corresponding estimating sample. These safeguards prevent rare numerical instabilities arising from extreme pilot residuals or near-zero predicted variances; in the simulations, they are activated in only a small fraction of replications, and the qualitative conclusions are not driven by them.

\paragraph{Independent-sampling estimators.}
Using the independent sample $S_p^{\mathrm{ind}}$, we compute the GREG estimator $\hat{Y}_{\text{GREG}}$ from Section~\ref{sec:2}, with first-order inclusion probabilities $\pi_i^{(\mathrm{ind})}$ and working covariate vector $\bx_i=(1,x_{1i},x_{2i})^\top$. Using $S_{np}$, we compute the propensity-based estimators $\hat{Y}_{\text{IPW}}$ and $\hat{Y}_{\text{DR}}$ from Section~\ref{sec:2} under a working MAR propensity score model with the same covariate vector. Finally, we compute the fusion estimator $\hat{Y}_{\text{GREG-DR},\alpha}$ in \eqref{eq:fusion} with the simple size-based rule $\alpha=|S_p^{\mathrm{ind}}|/(|S_p^{\mathrm{ind}}|+|S_{np}|)$. This rule is widely used in practice but is not variance-minimizing; the results below should therefore be read as a comparison against a typical implementation of the fusion estimator rather than against its theoretically optimal version.

In the NMAR setup, the propensity score model is misspecified by construction, which is intended to illustrate the sensitivity of $\hat{Y}_{\text{IPW}}$, $\hat{Y}_{\text{DR}}$, and $\hat{Y}_{\text{GREG-DR},\alpha}$ to violations of the MAR assumption.

\subsection{Performance measures}\label{sec:6:perf}

Let $\hat{Y}^{(r)}$ denote an estimator computed in Monte Carlo replication $r=1,\dots,R$. We use $R=100\,000$. Since the population is fully generated, the true total $Y=\sum_{i\in U}y_i$ is known. We report the relative bias (RB) and the relative root mean squared error (RRMSE) in percent,
\begin{equation*}
\mathrm{RB}(\hat{Y})=100\times\frac{\bar{\hat{Y}}-Y}{Y} \quad \text{and} \quad \mathrm{RRMSE}(\hat{Y})=100\times\frac{\left\{R^{-1}\sum_{r=1}^R\big(\hat{Y}^{(r)}-Y\big)^2\right\}^{1/2}}{Y},
\end{equation*}
where $\bar{\hat{Y}}=R^{-1}\sum_{r=1}^R\hat{Y}^{(r)}$.
For estimators equipped with a design-based variance estimator $\hat{V}^{(r)}$ (the sequential estimators and the GREG estimator $\hat{Y}_{\text{GREG}}$), we report the ratio $\bar{V}/V_{\mathrm{MC}}(\hat{Y})$, where
\begin{equation*}
\bar{V}=R^{-1}\sum_{r=1}^R\hat{V}^{(r)} \quad \text{and} \quad V_{\mathrm{MC}}(\hat{Y})=(R-1)^{-1}\sum_{r=1}^R\big(\hat{Y}^{(r)}-\bar{\hat{Y}}\big)^2.
\end{equation*}
We also report the coverage rate of nominal $95\%$ Wald intervals $\hat{Y}^{(r)}\pm z_{0.975}\{\hat{V}^{(r)}\}^{1/2}$, where $z_{0.975}$ denotes the $97.5$th quantile of the standard normal distribution. For the propensity-based estimators and the fusion estimator, we focus on RB and RRMSE and do not include variance estimation in the comparisons.

\subsection{Main results under the MAR mechanism}\label{sec:6:mar}

Table~\ref{tab:sim_mar} summarizes Monte Carlo performance measures under the MAR mechanism. Most estimators are essentially unbiased, with the main exception of the PPS version of $\hat{Y}_{\text{sepDI}}(q^{(\sigma)})$. Within the sequential class under the estimated optimal design, $\hat{Y}_{\text{sepDI}}(q^{(b)})$ and $\hat{Y}_{\text{comDI}}(q^{(\sigma)})$ attain the smallest RRMSE, while $\hat{Y}_{\text{sepDI}}(q^{(\sigma)})$ and $\hat{Y}_{\text{adDI}}(q^{(\sigma)})$ are only slightly less efficient. The GREG estimator is substantially less efficient than the sequential estimators. The IPW, DR, and GREG-DR estimators also perform well under MAR, as expected, but none improve on the best sequential regression estimators.

For $\hat{Y}_{\text{sepDI}}(q^{(\sigma)})$, the equal-Poisson design incurs only a modest loss of efficiency relative to the estimated optimal design, whereas the PPS design is unstable. Because the boxplots in Figure~\ref{fig:sim_mar} suppress outlying replications for readability, the PPS boxplot mainly reflects the central mass of the error distribution; the inflated RRMSE in Table~\ref{tab:sim_mar} is driven by a small number of extreme realizations.

Coverage of the nominal $95\%$ Wald intervals for the sequential estimators is between $0.947$ and $0.950$, and the variance-estimator ratios $\bar{V}/V_{\mathrm{MC}}$ are close to one, indicating well-calibrated variance estimation. The slight downward deviation from $0.950$ is within Monte Carlo noise at $R = 100\,000$.

\begin{table}[!p]
\centering
\small
\caption{Monte Carlo performance measures under the MAR mechanism. The Design column indicates the second-stage Poisson design on $U_1$ for the sequential estimators: ``optimal'' uses $\pi_i^{(p)}\propto \hat{\sigma}_i$, ``equal'' uses $\pi_i^{(p)}\equiv n_p/N_1$, and ``PPS'' uses $\pi_i^{(p)}\propto x_{1i}$. RB and RRMSE are expressed as percentages of $Y$.}
\label{tab:sim_mar}
\begin{tabular}{llrrrr}
\hline
Estimator & Design & RB (\%) & RRMSE (\%) & $\bar{V}/V_{\mathrm{MC}}$ & Coverage \\
\hline
$\hat{Y}_{\text{DI}}$ & optimal & 0.004 & 0.477 & 0.996 & 0.950 \\
$\hat{Y}_{\text{HT}}^{\text{seq}}$ & optimal & 0.002 & 0.754 & 1.000 & 0.949 \\
$\hat{Y}_{\text{sepDI}}(q^{(b)})$ & optimal & 0.000 & 0.418 & 0.990 & 0.948 \\
$\hat{Y}_{\text{sepDI}}(q^{(\sigma)})$ & optimal & 0.000 & 0.437 & 0.991 & 0.948 \\
$\hat{Y}_{\text{sepDI}}(q^{(\sigma)})$ & equal & 0.002 & 0.451 & 0.988 & 0.947 \\
$\hat{Y}_{\text{sepDI}}(q^{(\sigma)})$ & PPS & -0.234 & 109.239 & 0.575 & 0.925 \\
$\hat{Y}_{\text{comDI}}(q^{(\sigma)})$ & optimal & 0.001 & 0.417 & 0.996 & 0.949 \\
$\hat{Y}_{\text{adDI}}(q^{(\sigma)})$ & optimal & 0.000 & 0.435 & 0.992 & 0.948 \\
$\hat{Y}_{\text{GREG}}$ &  & -0.005 & 1.863 & 0.995 & 0.948 \\
$\hat{Y}_{\text{IPW}}$ &  & 0.000 & 0.467 &  &  \\
$\hat{Y}_{\text{DR}}$ &  & 0.000 & 0.465 &  &  \\
$\hat{Y}_{\text{GREG-DR},\alpha}$ &  & -0.000 & 0.482 &  &  \\
\hline
\end{tabular}
\end{table}

\begin{figure}[!p]
\centering
\includegraphics[width=0.95\textwidth]{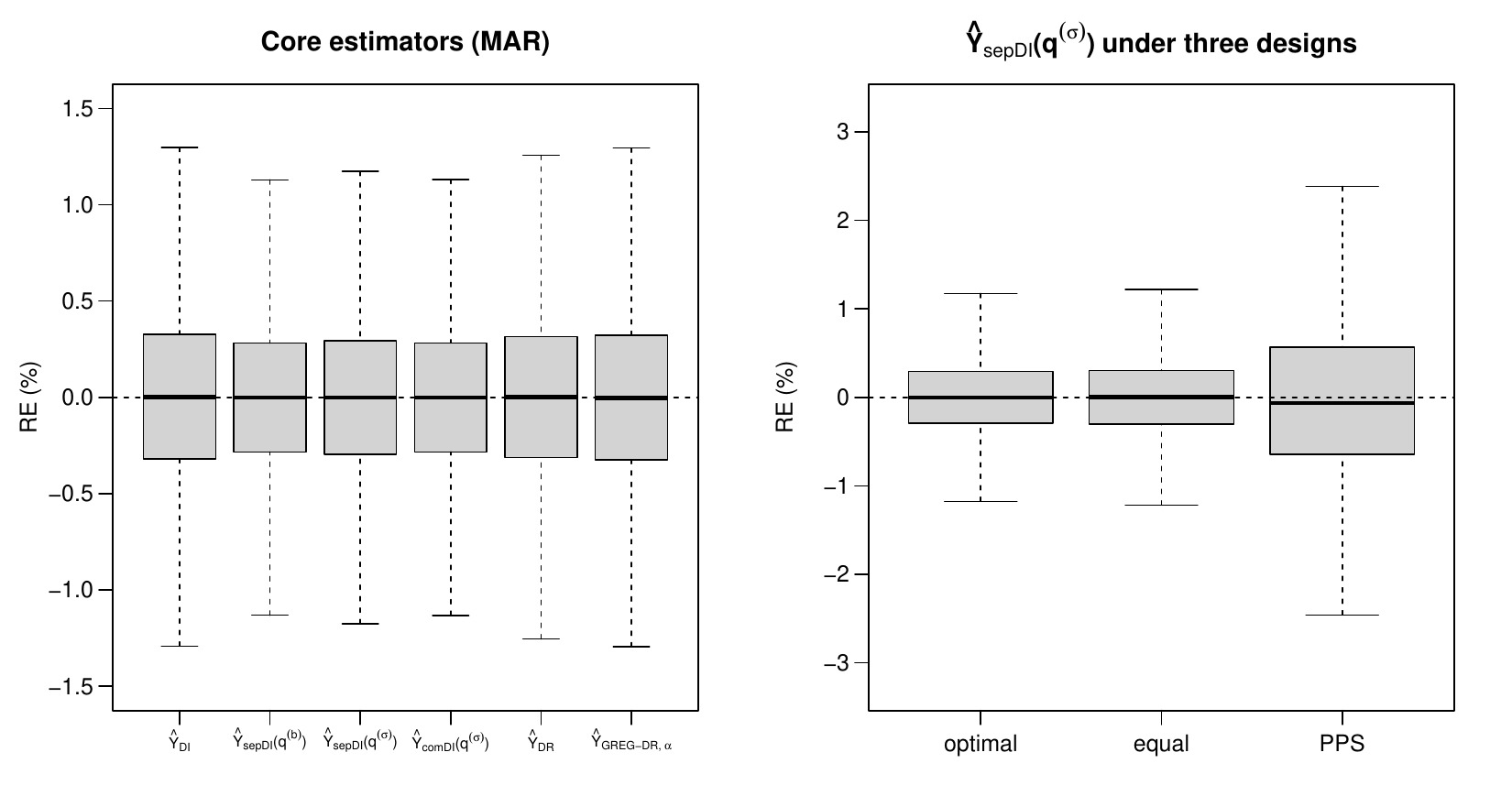}
\caption{Boxplots of relative errors $\mathrm{RE}=100(\hat{Y}-Y)/Y$ (in percent) across Monte Carlo replications under the MAR mechanism. Outlier points are suppressed. For readability, the vertical axis limits are truncated at the $0.999$ quantile of $|\mathrm{RE}|$ in the left panel and at the $0.99$ quantile in the right panel (computed within each panel), so replications outside the displayed range are not shown.}
\label{fig:sim_mar}
\end{figure}

The left panel of Figure~\ref{fig:sim_mar} compares selected estimators: the sequential estimators $\hat{Y}_{\text{DI}}$, $\hat{Y}_{\text{sepDI}}(q^{(b)})$, $\hat{Y}_{\text{sepDI}}(q^{(\sigma)})$, and $\hat{Y}_{\text{comDI}}(q^{(\sigma)})$ are shown under the estimated optimal design, while $\hat{Y}_{\text{DR}}$ and $\hat{Y}_{\text{GREG-DR},\alpha}$ are included as benchmark competitors. The right panel highlights the effect of the second-stage design on $\hat{Y}_{\text{sepDI}}(q^{(\sigma)})$. The display confirms that the selected estimators in the left panel are centered close to zero and that the equal-Poisson design remains reasonably competitive, while PPS exhibits a much wider spread.
\subsection{Main results under the NMAR mechanism}\label{sec:6:nmar}

Table~\ref{tab:sim_nmar} reports the corresponding results under the NMAR mechanism. The sequential estimators remain essentially unbiased, with variance estimates close to the empirical Monte Carlo variance and coverage close to the nominal $95\%$ level. Within the estimated optimal design, $\hat{Y}_{\text{sepDI}}(q^{(b)})$ attains the smallest RRMSE, closely followed by $\hat{Y}_{\text{comDI}}(q^{(\sigma)})$ and $\hat{Y}_{\text{sepDI}}(q^{(\sigma)})$. The GREG estimator remains essentially unbiased but is substantially less efficient than the sequential estimators.

In contrast, the IPW and DR estimators become severely biased when the propensity model is misspecified under NMAR, with relative biases exceeding $4.5\%$. The GREG-DR fusion estimator, although less biased than IPW and DR, remains substantially less accurate than the sequential estimators. As under MAR, the equal-Poisson design yields only a modest efficiency loss relative to the estimated optimal design, whereas PPS remains clearly inefficient.

\begin{table}[!ht]
\centering
\small
\caption{Monte Carlo performance measures under the NMAR mechanism. The Design column indicates the second-stage Poisson design on $U_1$ for the sequential estimators: ``optimal'' uses $\pi_i^{(p)}\propto \hat{\sigma}_i$, ``equal'' uses $\pi_i^{(p)}\equiv n_p/N_1$, and ``PPS'' uses $\pi_i^{(p)}\propto x_{1i}$. RB and RRMSE are expressed as percentages of $Y$.}
\label{tab:sim_nmar}
\begin{tabular}{llrrrr}
\hline
Estimator & Design & RB (\%) & RRMSE (\%) & $\bar{V}/V_{\mathrm{MC}}$ & Coverage \\
\hline
$\hat{Y}_{\text{DI}}$ & optimal & 0.001 & 0.406 & 1.003 & 0.950 \\
$\hat{Y}_{\text{HT}}^{\text{seq}}$ & optimal & -0.004 & 0.649 & 1.007 & 0.950 \\
$\hat{Y}_{\text{sepDI}}(q^{(b)})$ & optimal & -0.002 & 0.356 & 0.998 & 0.949 \\
$\hat{Y}_{\text{sepDI}}(q^{(\sigma)})$ & optimal & -0.002 & 0.369 & 0.998 & 0.950 \\
$\hat{Y}_{\text{sepDI}}(q^{(\sigma)})$ & equal & -0.001 & 0.386 & 0.991 & 0.947 \\
$\hat{Y}_{\text{sepDI}}(q^{(\sigma)})$ & PPS & -0.008 & 4.844 & 0.865 & 0.928 \\
$\hat{Y}_{\text{comDI}}(q^{(\sigma)})$ & optimal & 0.000 & 0.365 & 1.003 & 0.951 \\
$\hat{Y}_{\text{adDI}}(q^{(\sigma)})$ & optimal & -0.002 & 0.369 & 0.997 & 0.950 \\
$\hat{Y}_{\text{GREG}}$ &  & -0.010 & 1.861 & 0.996 & 0.947 \\
$\hat{Y}_{\text{IPW}}$ &  & 4.566 & 4.587 &  &  \\
$\hat{Y}_{\text{DR}}$ &  & 4.580 & 4.600 &  &  \\
$\hat{Y}_{\text{GREG-DR},\alpha}$ &  & 3.908 & 3.935 &  &  \\
\hline
\end{tabular}
\end{table}

\begin{figure}[!ht]
\centering
\includegraphics[width=0.95\textwidth]{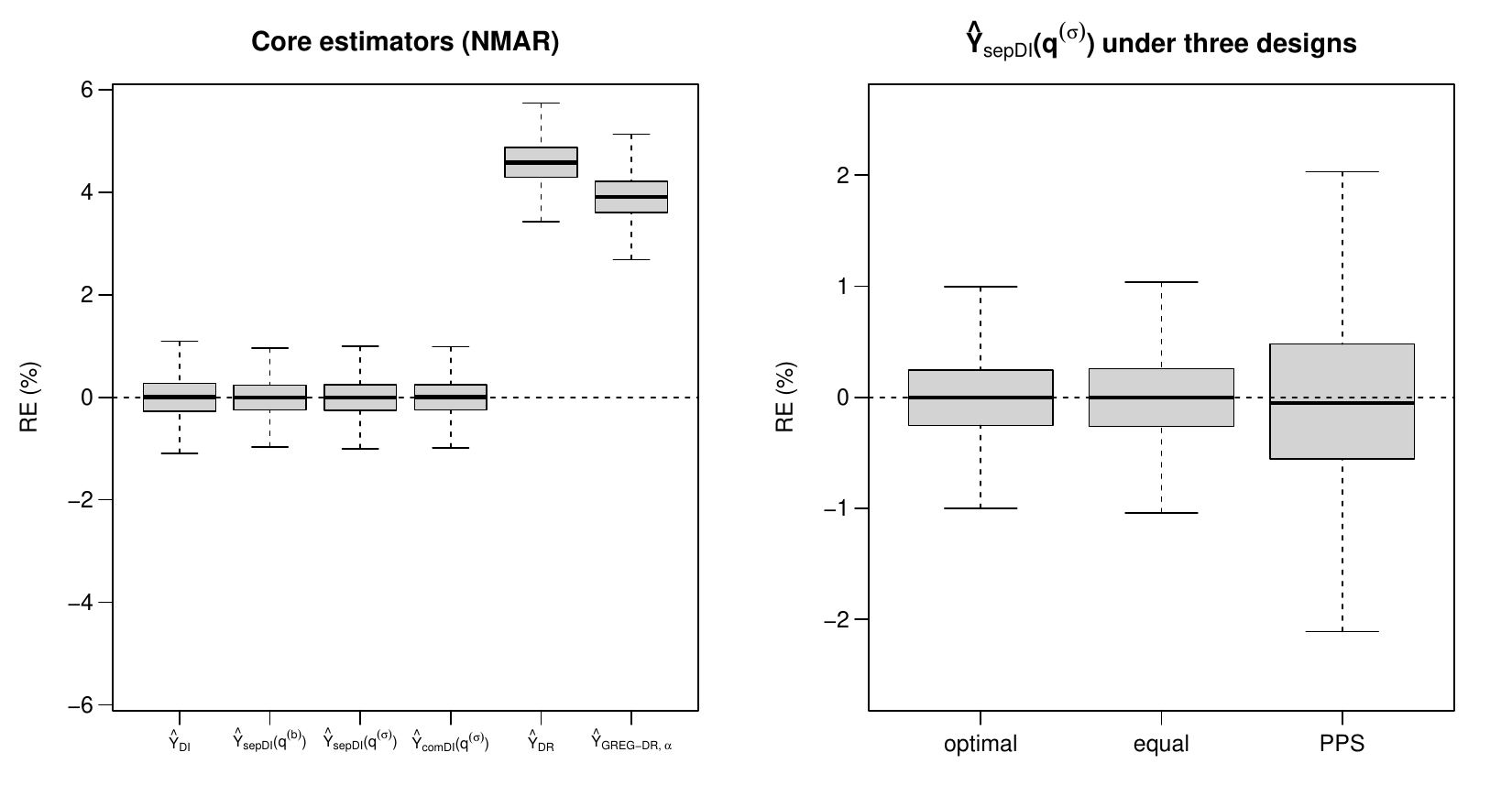}
\caption{Boxplots of relative errors $\mathrm{RE}=100(\hat{Y}-Y)/Y$ (in percent) across Monte Carlo replications under the NMAR mechanism. Outlier points are suppressed. For readability, the vertical axis limits are truncated at the $0.999$ quantile of $|\mathrm{RE}|$ in the left panel and at the $0.99$ quantile in the right panel (computed within each panel), so replications outside the displayed range are not shown.}
\label{fig:sim_nmar}
\end{figure}

In the left panel of Figure~\ref{fig:sim_nmar}, the sequential estimators remain concentrated around zero, while the DR estimator and the GREG-DR fusion estimator are visibly shifted because of the NMAR bias. The right panel again shows that the equal-Poisson design stays close to the estimated optimal design for $\hat{Y}_{\text{sepDI}}(q^{(\sigma)})$, whereas PPS leads to markedly larger dispersion.

\subsection{Homogeneity test and adaptive choice}\label{sec:6:test}

Table~\ref{tab:sim_test} summarizes the coefficient homogeneity test at the $\alpha=0.05$ level under both mechanisms. The rejection rate is $0.842$ under MAR and $0.967$ under NMAR.

The rejection rate of $0.842$ under MAR may at first appear surprising and warrants discussion. With $N=10\,000$, $n_{np}\approx 7\,000$, and $n_p\approx 1\,200$, the sample sizes used to estimate $\bbeta_{np}$ and $\bbeta_p$ are large enough that the test detects even very small differences between the two coefficient vectors. In finite populations, exact equality $\bbeta_{np}=\bbeta_p$ is essentially never the truth, even when the underlying mechanism is MAR; what holds under MAR together with correct specification of the mean model is asymptotic equality, and at moderate-to-large sample sizes the test will reject substantively negligible differences.

Two practical implications follow. First, the homogeneity test should be read as a sensitivity diagnostic rather than as a binary rule for choosing between $\hat{Y}_{\text{sepDI}}$ and $\hat{Y}_{\text{comDI}}$. The test answers ``is there evidence of any coefficient heterogeneity?'', whereas the relevant practical question is ``is the heterogeneity large enough to outweigh the variance reduction from pooling?''. Second, the MSE comparison between $\hat{Y}_{\text{sepDI}}$ and $\hat{Y}_{\text{comDI}}$ depends not only on the size of the coefficient differences but also on how different the auxiliary totals are between $S_{np}$ and $U_1$: when the corresponding $\bx$-totals are similar, even nontrivial coefficient differences contribute little to the MSE of the combined estimator.

This is exactly the pattern observed in Tables~\ref{tab:sim_mar} and~\ref{tab:sim_nmar}: despite frequent rejection, $\hat{Y}_{\text{comDI}}(q^{(\sigma)})$ is slightly more efficient than $\hat{Y}_{\text{sepDI}}(q^{(\sigma)})$ under MAR, and only marginally less efficient under NMAR.

Consequently, the adaptive estimator $\hat{Y}_{\text{adDI}}(q^{(\sigma)})$ defined in \eqref{eq:addi_rule} is nearly indistinguishable from $\hat{Y}_{\text{sepDI}}(q^{(\sigma)})$ in both tables, since the test rejects in the vast majority of replications. The adaptive rule is conservative by design: it defaults to the safer separate estimator whenever any heterogeneity is detected. A more refined diagnostic---one that compares an estimated MSE difference against a tolerance scaled by the auxiliary-total contrast---could, in principle, recover the modest efficiency gain offered by the combined estimator under MAR, but we leave this to future work.

\begin{table}[!ht]
\centering
\caption{Summary of the coefficient homogeneity test across Monte Carlo replications under the MAR and NMAR mechanisms.}
\label{tab:sim_test}
\begin{tabular}{llrrrr}
\hline
Mechanism & $R$ & $\alpha$ & Reject rate & Mean($p$-value) & Median($p$-value) \\
\hline
MAR  & 100\,000 & 0.05 & 0.84177 & 0.055346 & 0.000001 \\
NMAR & 100\,000 & 0.05 & 0.96679 & 0.012085 & 0.000000 \\
\hline
\end{tabular}
\end{table}

\section{Real-data applications}\label{sec:7}

This section presents real-data Monte Carlo experiments based on data provided by the State Data Agency (Statistics Lithuania). Two applications are considered: turnover for service enterprises (Section~\ref{sec:7.1}) and drug sales for pharmacies (Section~\ref{sec:7.2}). Both follow the sequential framework of Section~\ref{sec:3}: an administrative source or a voluntary sample defines the realized $S_{np}$, and the second-stage probability sample is drawn from $U_1=U\setminus S_{np}$.

The two applications differ in the degree of heterogeneity between $S_{np}$ and $U_1$, providing complementary stress tests of the proposed framework. The turnover application features strong heterogeneity: an administrative source over-represents large enterprises relative to the population. The pharmacy application features weak heterogeneity: a voluntary sample of large pharmacy chains overlaps substantially with the rest of the population on the relevant scale. As anticipated by the discussion in Section~\ref{sec:6:test}, the separate regression estimator is preferable in the heterogeneous setting, while the combined estimator offers a small efficiency gain in the more homogeneous one.

\subsection*{Framework for Monte Carlo experiments}

In each real-data experiment, we condition on the realized $S_{np}$ and repeatedly draw only the second-stage Poisson sample $S_p\subset U_1$, with $R=100\,000$ Monte Carlo replications. The first-order inclusion probabilities $\pi_i^{(p)}$ are computed from the pilot variance model of Section~\ref{sec:4} via \eqref{eq:est_opt_pi0}, and the fitted variance proxy $\hat{\sigma}_i^2$ from the pilot model is used both in $\pi_i^{(p)}$ and in the variance-weighted specification $q_i^{(\sigma)}$. The reported performance measures therefore describe design-based repeated-sampling behavior with respect to the second-stage design, conditional on the observed $S_{np}$.

Under this conditional setup, estimators that depend only on $S_{np}$---such as the IPW and DR estimators of Section~\ref{sec:2}---do not exhibit Monte Carlo variability and are not included in the comparison. We focus on the sequential design-based estimators introduced in Section~\ref{sec:6:estimators}: $\hat{Y}_{\text{DI}}$, the sequential HT estimator $\hat{Y}_{\text{HT}}^{\text{seq}}$, $\hat{Y}_{\text{sepDI}}(q^{(b)})$, $\hat{Y}_{\text{sepDI}}(q^{(\sigma)})$, $\hat{Y}_{\text{comDI}}(q^{(\sigma)})$, and the adaptive estimator $\hat{Y}_{\text{adDI}}(q^{(\sigma)})$ defined in \eqref{eq:addi_rule}.

\subsection{Turnover data for service enterprises}\label{sec:7.1}

\subsubsection{Population, variables, and sampling}

We consider a finite pseudo-real population of $N=2\,839$ Lithuanian service enterprises. The study variable $y$ is the monthly turnover. Auxiliary information is available for all units from administrative registers; in this application, we use a single covariate $x_1$ (annual turnover) and set $\bx_i=(1,x_{1i})^\top$.

The realized non-probability sample $S_{np}$ is obtained from administrative value-added tax (VAT) records and contains $n_{np}=2\,232$ units. Thus, the complementary stratum has size $N_1=607$. Because small enterprises are often not VAT payers, the administrative source and its complement are markedly heterogeneous, which makes this application a useful stress test of the sequential design-based strategy.

In the second stage, we set the expected Poisson sample size to $n_p=242$ (about $40\%$ of $U_1$). The inclusion probabilities on $U_1$ are proportional to $\hat{\sigma}_i$ and scaled to satisfy $\sum_{i\in U_1}\pi_i^{(p)}=n_p$.

\subsubsection[Diagnostics: heterogeneity between Snp and U1]{Diagnostics: heterogeneity between $S_{np}$ and $U_1$}

Figure~\ref{fig:vat_heterogeneity} illustrates the lack of representativeness of the administrative pilot by comparing the outcome distribution and the $y$--$x_1$ relationship across $S_{np}$ and $U_1$. For readability, we apply the transformation $t\mapsto \log(1+t)$, which reduces the impact of heavy tails commonly encountered in turnover data. The density plot shows strong separation between the two groups, and the scatter plot of $\log(1+y)$ versus $\log(1+x_1)$ highlights that $S_{np}$ over-represents large units relative to $U_1$. Such heterogeneity is precisely the setting in which approaches relying on MAR and common support assumptions can be fragile, motivating the sequential design-based strategy adopted here.

\begin{figure}[!ht]
\centering
\begin{subfigure}{0.48\textwidth}
\centering
\includegraphics[width=\textwidth]{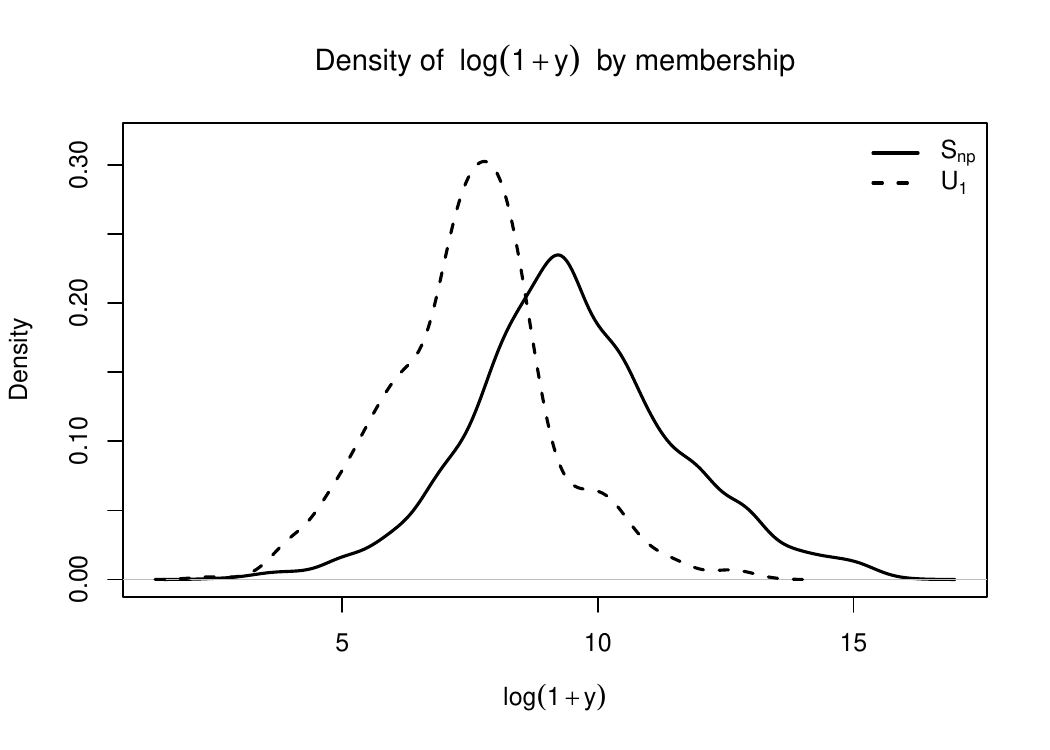}
\caption{Density of $\log(1+y)$ by sample membership.}
\end{subfigure}\hfill
\begin{subfigure}{0.48\textwidth}
\centering
\includegraphics[width=\textwidth]{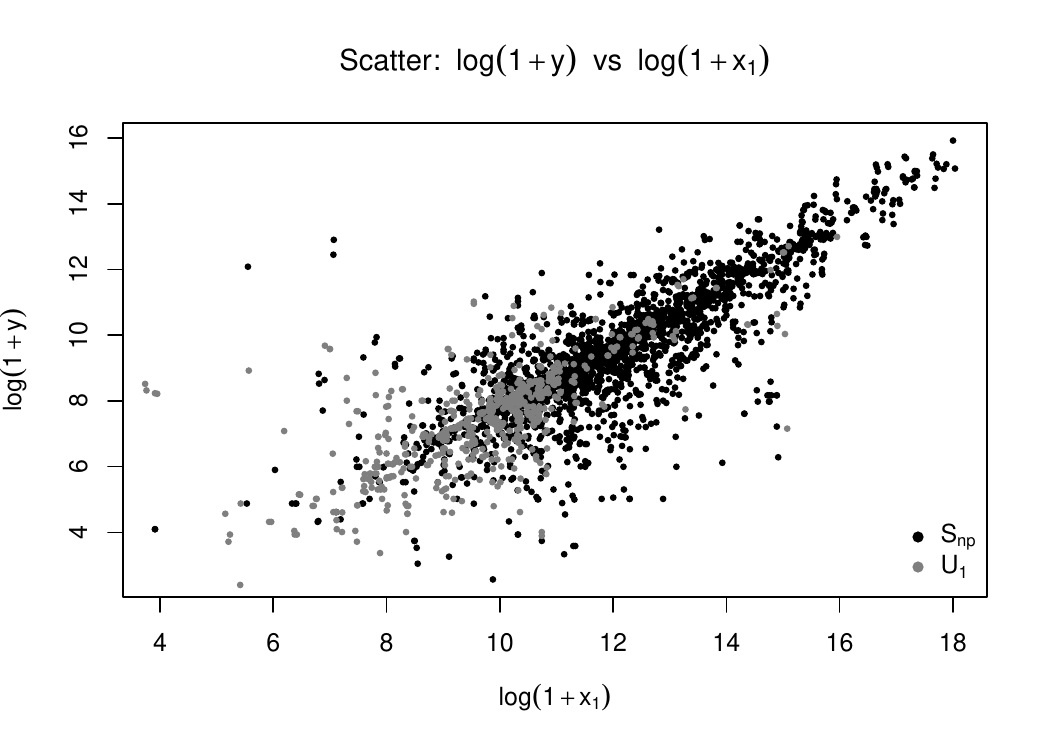}
\caption{Scatter of $\log(1+y)$ versus $\log(1+x_1)$.}
\end{subfigure}
\caption{Heterogeneity between the administrative pilot sample $S_{np}$ and its complement $U_1$.}
\label{fig:vat_heterogeneity}
\end{figure}

\subsubsection{Main results}

Table~\ref{tab:vat_mc_table_paper} summarizes the Monte Carlo performance of the core estimators under the estimated optimal design with $\pi_i^{(p)}\propto \hat{\sigma}_i$. The table also includes two additional rows that evaluate $\hat{Y}_{\text{sepDI}}(q^{(\sigma)})$ under two alternative Poisson designs on $U_1$: an equal-probability design with $\pi_i^{(p)}\equiv n_p/N_1$ and a Poisson probability proportional to size (PPS) design with $\pi_i^{(p)}\propto x_{1i}$. Relative to the estimated optimal design, both alternatives lead to a substantial loss of efficiency for $\hat{Y}_{\text{sepDI}}(q^{(\sigma)})$. In this application, the separate regression estimators $\hat{Y}_{\text{sepDI}}(q^{(b)})$ and $\hat{Y}_{\text{sepDI}}(q^{(\sigma)})$ attain the smallest RRMSE. The combined estimator $\hat{Y}_{\text{comDI}}(q^{(\sigma)})$ is less efficient, which is consistent with the strong heterogeneity between $S_{np}$ and $U_1$. For the core estimators, the variance estimator accuracy, summarized by $\bar{V}/V_{\mathrm{MC}}(\hat{Y})$, is close to one and the resulting Wald intervals have coverage reasonably close to the nominal $95\%$ level.

\begin{table}[!ht]
\centering
\small
\caption{Monte Carlo performance measures for the turnover application. The Design column indicates the second-stage Poisson design on $U_1$: ``optimal'' uses $\pi_i^{(p)}\propto \hat{\sigma}_i$, ``equal'' uses $\pi_i^{(p)}\equiv n_p/N_1$, and ``PPS'' uses $\pi_i^{(p)}\propto x_{1i}$. RB and RRMSE are expressed as percentages of $Y$.}
\label{tab:vat_mc_table_paper}
\begin{tabular}{llrrrr}
\hline
Estimator & Design & RB (\%) & RRMSE (\%) & $\bar{V}/V_{\mathrm{MC}}$ & Coverage \\
\hline
$\hat{Y}_{\text{DI}}$ & optimal & 0.003 & 0.096 & 0.988 & 0.941 \\
$\hat{Y}_{\text{HT}}^{\text{seq}}$ & optimal & 0.000 & 0.088 & 1.001 & 0.935 \\
$\hat{Y}_{\text{sepDI}}(q^{(b)})$ & optimal & 0.001 & 0.075 & 0.984 & 0.923 \\
$\hat{Y}_{\text{sepDI}}(q^{(\sigma)})$ & optimal & 0.001 & 0.076 & 0.982 & 0.923 \\
$\hat{Y}_{\text{sepDI}}(q^{(\sigma)})$ & equal & -0.003 & 0.239 & 0.592 & 0.827 \\
$\hat{Y}_{\text{sepDI}}(q^{(\sigma)})$ & PPS & 0.008 & 1.732 & 1.090 & 0.906 \\
$\hat{Y}_{\text{comDI}}(q^{(\sigma)})$ & optimal & 0.003 & 0.097 & 1.019 & 0.947 \\
$\hat{Y}_{\text{adDI}}(q^{(\sigma)})$ & optimal & 0.001 & 0.076 & 0.982 & 0.923 \\
\hline
\end{tabular}
\end{table}
Figure~\ref{fig:vat_boxplot} displays boxplots of relative errors (in percent) across Monte Carlo replications. The left panel compares the core estimators under the estimated optimal design, and the right panel highlights the effect of the probability sampling design in the second stage for $\hat{Y}_{\text{sepDI}}(q^{(\sigma)})$ under the three Poisson designs.

\begin{figure}[!ht]
\centering
\includegraphics[width=0.95\textwidth]{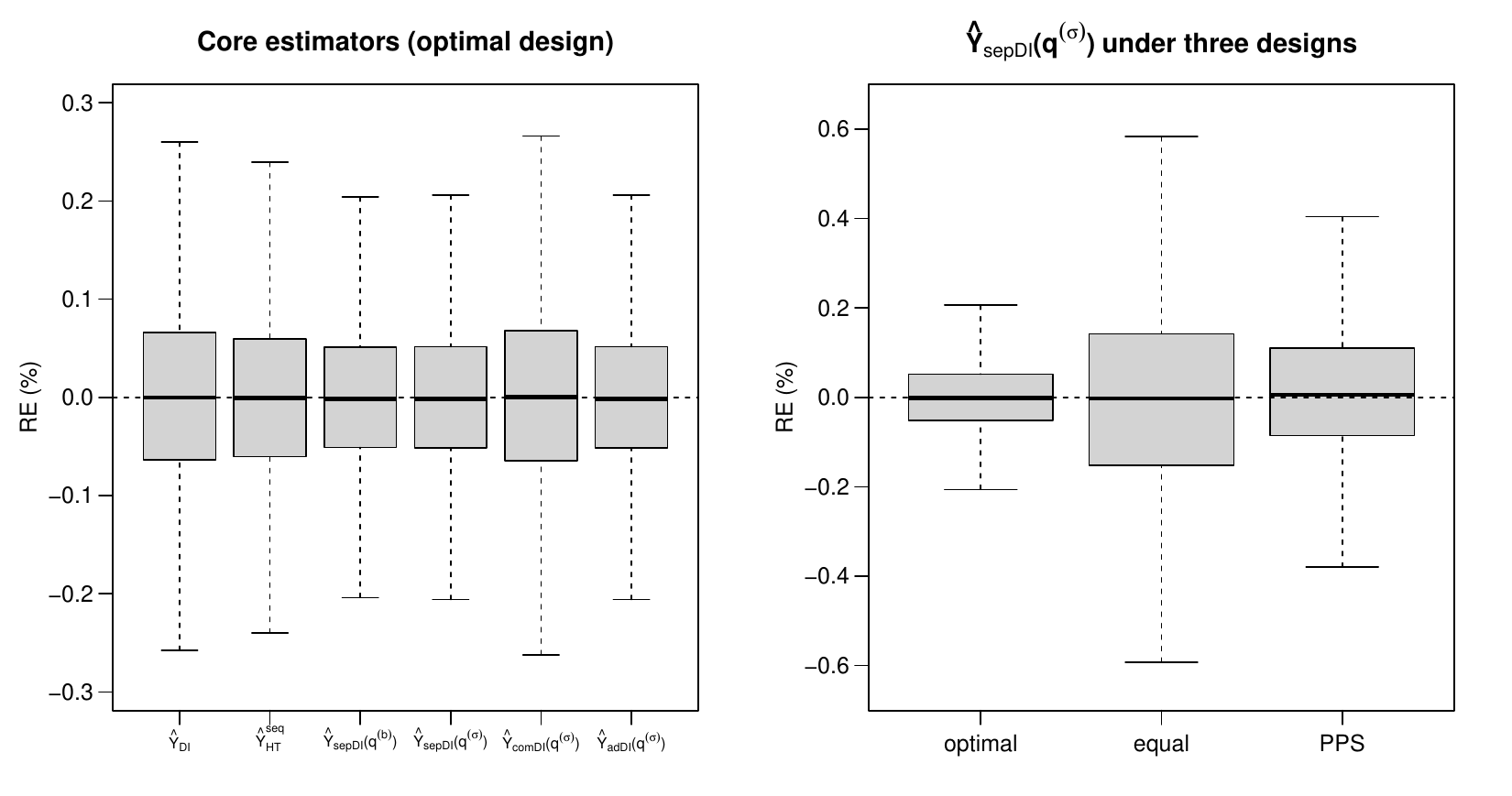}
\caption{Boxplots of relative errors $\mathrm{RE}=100(\hat{Y}-Y)/Y$ (in percent) across Monte Carlo replications. Outlier points are suppressed. For readability, the vertical axis limits are truncated at the $0.999$ quantile of $|\mathrm{RE}|$ in the left panel and at the $0.99$ quantile in the right panel (computed within each panel), so replications outside the displayed range are not shown.}
\label{fig:vat_boxplot}
\end{figure}

\subsubsection{Homogeneity test and adaptive choice}

Table~\ref{tab:vat_test_summary} reports the outcome of the coefficient homogeneity test at the $\alpha=0.05$ level. The rejection rate is essentially one, indicating strong evidence against coefficient homogeneity between $S_{np}$ and $U_1$ in this application. Accordingly, the adaptive estimator $\hat{Y}_{\text{adDI}}(q^{(\sigma)})$ in \eqref{eq:addi_rule} is nearly indistinguishable from the separate estimator $\hat{Y}_{\text{sepDI}}(q^{(\sigma)})$ in Table~\ref{tab:vat_mc_table_paper}.

\begin{table}[!ht]
\centering
\caption{Summary of the coefficient homogeneity test across Monte Carlo replications.}
\label{tab:vat_test_summary}
\begin{tabular}{lrrrr}
\hline
$R$ & $\alpha$ & Reject rate & Mean($p$-value) & Median($p$-value) \\
\hline
100\,000 & 0.05 & 0.99998 & 0.000021 & 0.000000 \\
\hline
\end{tabular}
\end{table}

\subsection{Drug sales data for pharmacies}\label{sec:7.2}

\subsubsection{Population, variables, and sampling}

We consider a finite pseudo-real population of $N=2\,426$ Lithuanian retail pharmacies. The study variable $y$ is the quarterly expenditure for reimbursed prescription medicines, recorded in cash register (receipt) data. As auxiliary information, we use $x_1$, defined as the quarterly expenditure on reimbursed prescription medicines from the national e-prescription system. We set $\bx_i=(1,x_{1i})^\top$.

The realized non-probability (voluntary) sample $S_{np}$ consists of units belonging to the two largest pharmacy chains and contains $n_{np}=1\,458$ units. Thus, the complementary stratum has size $N_1=968$, and the subsequent probability sample is selected from $U_1=U\setminus S_{np}$. In the second stage, we set the expected Poisson sample size to $n_p=387$ (about $40\%$ of $U_1$) and compute inclusion probabilities on $U_1$ under the estimated optimal design, $\pi_i^{(p)}\propto \hat{\sigma}_i$, scaled to satisfy $\sum_{i\in U_1}\pi_i^{(p)}=n_p$.

\subsubsection[Diagnostics: heterogeneity between Snp and U1]{Diagnostics: heterogeneity between $S_{np}$ and $U_1$}

Figure~\ref{fig:drugs_heterogeneity} compares the outcome distribution and the $y$--$x_1$ relationship across $S_{np}$ and $U_1$. For readability, we apply the transformation $t\mapsto \log(1+t)$, as in Figure~\ref{fig:vat_heterogeneity}. In contrast to the turnover application in Section~\ref{sec:7.1}, the two groups exhibit substantially weaker heterogeneity: the density plots overlap to a much larger extent, and the scatter plot suggests a comparable association between $y$ and $x_1$ in both strata.

\begin{figure}[!ht]
\centering
\begin{subfigure}{0.48\textwidth}
\centering
\includegraphics[width=\textwidth]{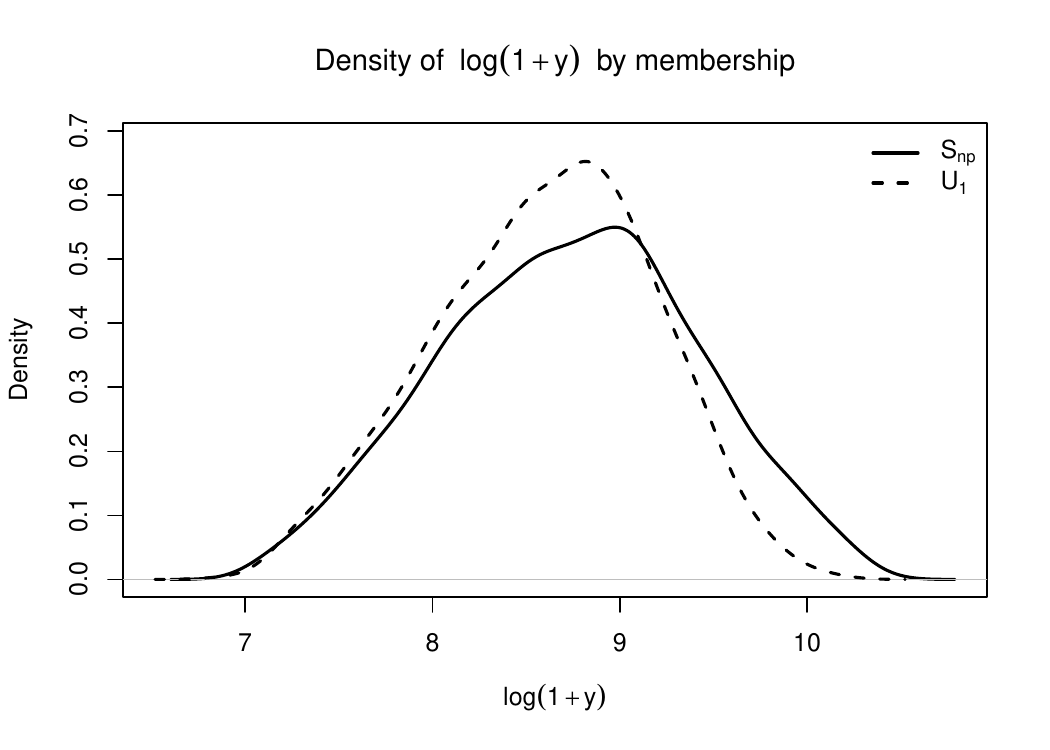}
\caption{Density of $\log(1+y)$ by sample membership.}
\end{subfigure}\hfill
\begin{subfigure}{0.48\textwidth}
\centering
\includegraphics[width=\textwidth]{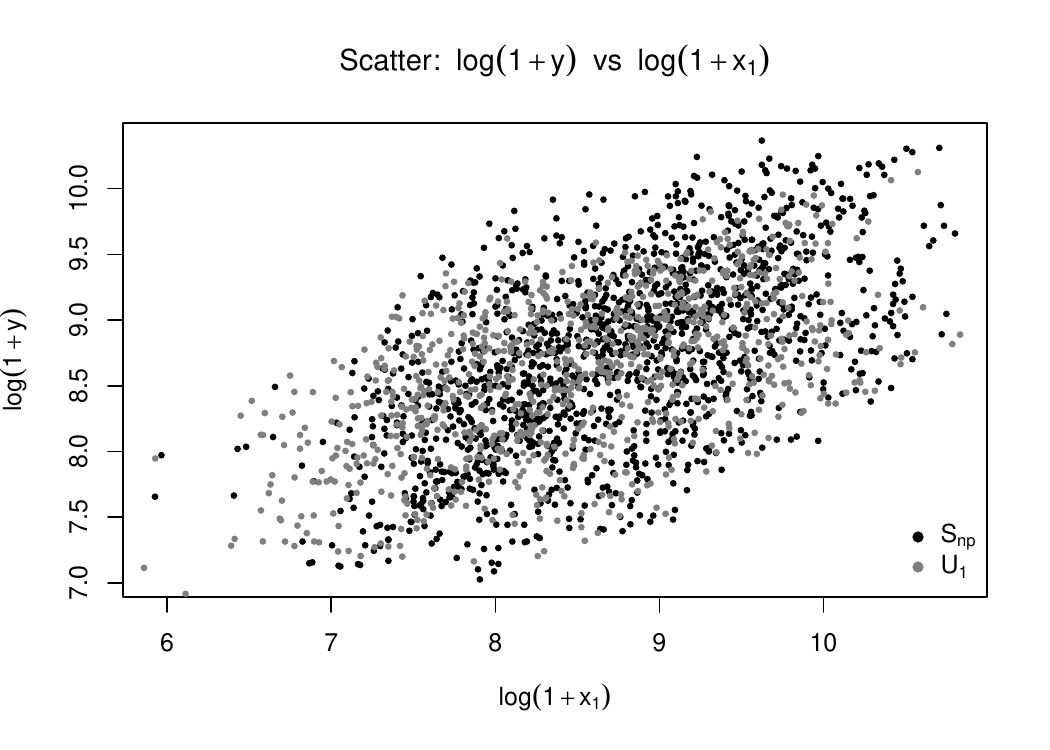}
\caption{Scatter of $\log(1+y)$ versus $\log(1+x_1)$.}
\end{subfigure}
\caption{Heterogeneity between the pharmacy pilot sample $S_{np}$ and its complement $U_1$.}
\label{fig:drugs_heterogeneity}
\end{figure}

\subsubsection{Main results}

Table~\ref{tab:drugs_mc_table_paper} reports Monte Carlo performance measures under the estimated optimal second-stage design with $\pi_i^{(p)}\propto \hat{\sigma}_i$. For $\hat{Y}_{\text{sepDI}}(q^{(\sigma)})$, we also include two alternative Poisson designs on $U_1$ (equal and PPS). Consistent with Figure~\ref{fig:drugs_heterogeneity}, the equal-Poisson design incurs only a modest loss of efficiency relative to the estimated optimal design, whereas the PPS design remains markedly inefficient. Within the optimal design, the regression-type estimators dominate $\hat{Y}_{\text{DI}}$ and $\hat{Y}_{\text{HT}}^{\text{seq}}$, with well-calibrated variance estimates and coverage close to the nominal $95\%$ level. Unlike in Section~\ref{sec:7.1}, the combined estimator $\hat{Y}_{\text{comDI}}(q^{(\sigma)})$ attains a slightly smaller RRMSE than $\hat{Y}_{\text{sepDI}}(q^{(\sigma)})$.

\begin{table}[!ht]
\centering
\small
\caption{Monte Carlo performance measures for the pharmacy application. The Design column indicates the second-stage Poisson design on $U_1$: ``optimal'' uses $\pi_i^{(p)}\propto \hat{\sigma}_i$, ``equal'' uses $\pi_i^{(p)}\equiv n_p/N_1$, and ``PPS'' uses $\pi_i^{(p)}\propto x_{1i}$. RB and RRMSE are expressed as percentages of $Y$.}
\label{tab:drugs_mc_table_paper}
\begin{tabular}{llrrrr}
\hline
Estimator & Design & RB (\%) & RRMSE (\%) & $\bar{V}/V_{\mathrm{MC}}$ & Coverage \\
\hline
$\hat{Y}_{\text{DI}}$ & optimal & 0.004 & 0.758 & 0.999 & 0.949 \\
$\hat{Y}_{\text{HT}}^{\text{seq}}$ & optimal & -0.001 & 1.502 & 1.005 & 0.950 \\
$\hat{Y}_{\text{sepDI}}(q^{(b)})$ & optimal & 0.002 & 0.670 & 0.996 & 0.948 \\
$\hat{Y}_{\text{sepDI}}(q^{(\sigma)})$ & optimal & 0.002 & 0.673 & 0.996 & 0.948 \\
$\hat{Y}_{\text{sepDI}}(q^{(\sigma)})$ & equal & 0.016 & 0.712 & 0.989 & 0.947 \\
$\hat{Y}_{\text{sepDI}}(q^{(\sigma)})$ & PPS & 0.013 & 3.164 & 1.306 & 0.939 \\
$\hat{Y}_{\text{comDI}}(q^{(\sigma)})$ & optimal & 0.001 & 0.656 & 1.003 & 0.950 \\
$\hat{Y}_{\text{adDI}}(q^{(\sigma)})$ & optimal & 0.002 & 0.673 & 0.996 & 0.948 \\
\hline
\end{tabular}
\end{table}

Figure~\ref{fig:drugs_boxplot} provides a graphical summary of relative errors; the display mirrors Figure~\ref{fig:vat_boxplot} to facilitate comparison across applications.

\begin{figure}[!ht]
\centering
\includegraphics[width=0.95\textwidth]{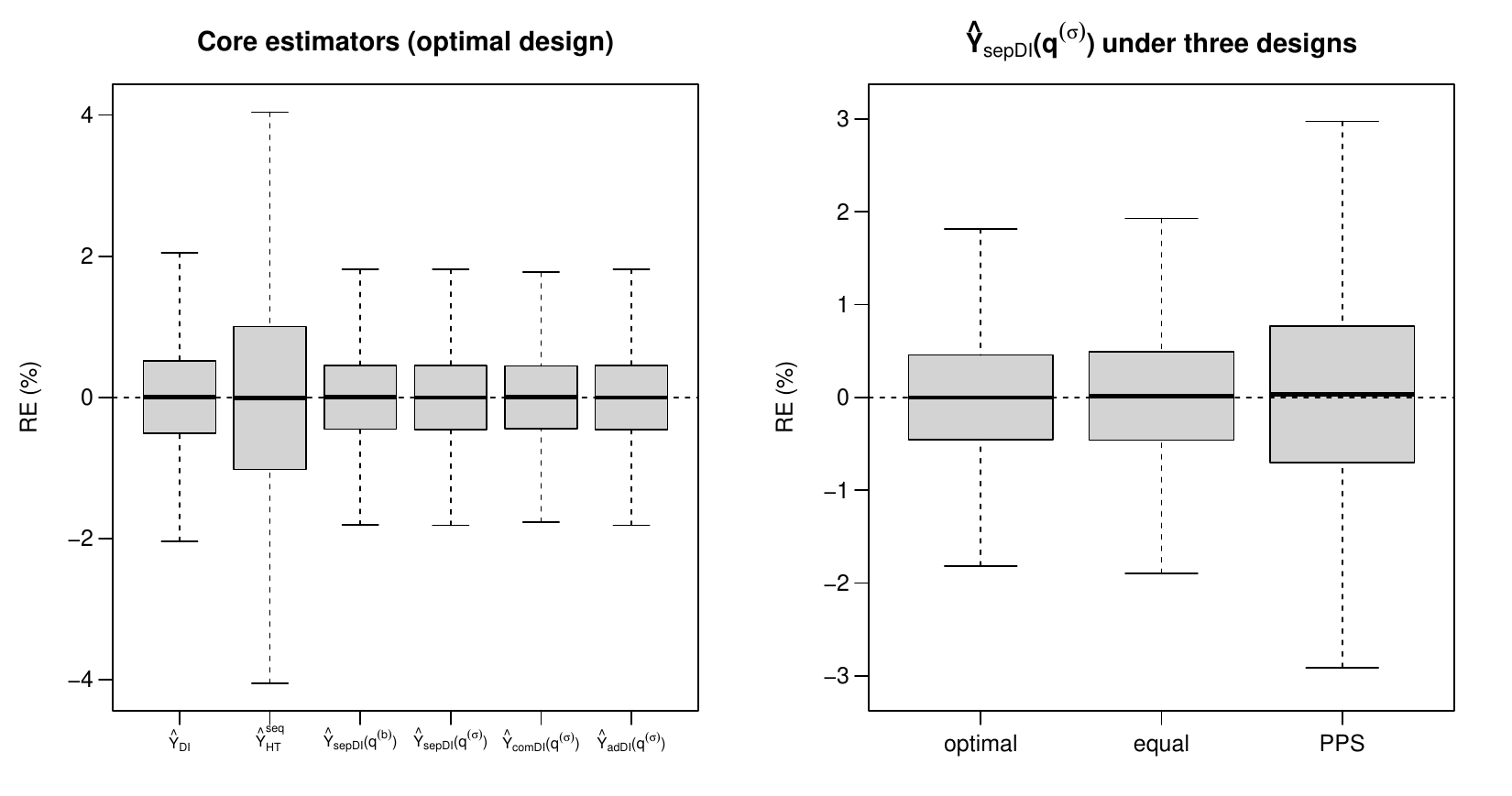}
\caption{Boxplots of relative errors $\mathrm{RE}=100(\hat{Y}-Y)/Y$ (in percent) across Monte Carlo replications. Outlier points are suppressed. For readability, the vertical axis limits are truncated at the $0.999$ quantile of $|\mathrm{RE}|$ in the left panel and at the $0.99$ quantile in the right panel (computed within each panel), so replications outside the displayed range are not shown.}
\label{fig:drugs_boxplot}
\end{figure}

\subsubsection{Homogeneity test and adaptive choice}

Table~\ref{tab:drugs_test_summary} reports the outcome of the coefficient homogeneity test at the $\alpha=0.05$ level. The rejection rate is close to one, and therefore the adaptive estimator $\hat{Y}_{\text{adDI}}(q^{(\sigma)})$ coincides with $\hat{Y}_{\text{sepDI}}(q^{(\sigma)})$ in this application. At the same time, Table~\ref{tab:drugs_mc_table_paper} shows that the combined estimator $\hat{Y}_{\text{comDI}}(q^{(\sigma)})$ achieves a modest RRMSE gain despite the statistical rejection of homogeneity. This is not contradictory: the test targets equality of regression coefficients, whereas the MSE difference between separate and combined estimation depends both on coefficient differences and on how different the auxiliary totals are between $S_{np}$ and $U_1$. In this application, the strata are fairly similar in terms of $x_1$, so even statistically detectable coefficient differences can have a limited practical impact on the overall MSE.

\begin{table}[!ht]
\centering
\caption{Summary of the coefficient homogeneity test across Monte Carlo replications.}
\label{tab:drugs_test_summary}
\begin{tabular}{lrrrr}
\hline
$R$ & $\alpha$ & Reject rate & Mean($p$-value) & Median($p$-value) \\
\hline
100\,000 & 0.05 & 0.99993 & 0.000213 & 0.000007 \\
\hline
\end{tabular}
\end{table}

\section{Conclusion}\label{sec:8}

A non-probability sample can be incorporated into finite-population inference by treating it as a pilot certainty stratum and drawing the probability sample only from the remaining units. This sequential framework leads to design-based data integration estimators that fit naturally into a standard stratified-sampling formulation, removing the need to model unknown selection probabilities and eliminating the tuning parameter required by fusion-type estimators. The resulting DI, separate GREG, and combined GREG estimators are transparent to interpret and straightforward to implement.

The paper develops two distinct theoretical contributions, layered rather than coupled. The first is a robustness result: both regression estimators are design-consistent under standard regularity conditions on the second-stage probability design, with no assumption whatsoever on the selection mechanism of the non-probability sample. Their plug-in variance estimators are likewise consistent. Design consistency, therefore, holds under MAR and NMAR alike, in contrast with inverse probability weighting and doubly robust estimators that lose this property under NMAR.

The second contribution is an efficiency result: under a working superpopulation model that holds in both strata, the pilot non-probability sample can be used to estimate the variance function and to construct second-stage inclusion probabilities that achieve Isaki--Fuller asymptotic optimality for the separate estimator. This optimality claim relies on assumptions strictly stronger than MAR. Crucially, however, its failure does not invalidate the inference; the design-consistency and variance-estimation results of the first contribution remain intact, so efficiency gains from the pilot-based design come at no cost to the validity of the inference.

The simulation study and real-data applications confirm both contributions. Under both MAR and NMAR mechanisms, the sequential estimators remain essentially unbiased, their variance estimators are well calibrated, and Wald intervals attain coverage close to the nominal $95\%$ level. The independent-sampling competitors that rely on MAR-type adjustments are competitive only when the MAR assumption holds and become seriously biased under NMAR. Within the sequential class, the regression-type estimators dominate the DI and stratified HT estimators, and the estimated optimal second-stage design is generally preferable to equal-probability and especially PPS sampling. The two real-data applications further illustrate that the practical choice between separate and combined regression depends on the heterogeneity between the pilot stratum and its complement: separate regression is the safer default in strongly heterogeneous settings (turnover application), whereas the combined estimator can offer a modest efficiency gain when the strata are similar (pharmacy application).

The coefficient homogeneity test is useful as a diagnostic, but its rejection event should be read as a sensitivity indicator rather than a binary selection rule between the separate and combined estimators. At the moderate-to-large sample sizes relevant for official statistics, the test rejects substantively negligible differences, and the resulting adaptive estimator typically defaults to the separate one. This is a conservative choice; refining the diagnostic to compare an estimated MSE difference against a tolerance scaled by the auxiliary-total contrast between $S_{np}$ and $U_1$ is a natural direction for future work.

\appendix

\section{Proofs}\label{app:proofs}

\begin{proposition}[Consistency of combined regression coefficient estimator]\label{prop:consistency-combined}
Let $N\to\infty$ with $n_p\to\infty$ and $n_{np}\to\infty$. Assume:
\begin{itemize}
  \item[\textnormal{(P1)}] The superpopulation mean model in \eqref{eq:superpop_model} holds in both population strata $S_{np}$ and $U_1$, i.e., $E_\xi(y_i\mid \bx_i)=\bx_i^\top\bbeta$. 
  \item[\textnormal{(P2)}] The covariates and outcomes satisfy
  \begin{equation*}
  \frac{1}{N}\sum_{i\in U}\|\bx_i\|^2<\infty \quad \text{and} \quad
  \frac{1}{N}\sum_{i\in U}y_i^2<\infty.
  \end{equation*}
  \item[\textnormal{(P3)}] There exists a constant $\pi_{\min}>0$ such that $\pi_i^{(p)}\ge \pi_{\min}$ for all $i\in U_1$, and $0<q_{\min}\le q_i\le q_{\max}<\infty$ for all $i\in U$.
  \item[\textnormal{(P4)}] For every component $z_i$ of $\pi_i^{(p)}q_i\bx_i\bx_i^\top$ and $\pi_i^{(p)}q_i\bx_i y_i$,
  \begin{equation*}
  \frac{1}{N_1}\sum_{i\in S_p}\frac{z_i}{\pi_i^{(p)}}=
  \frac{1}{N_1}\sum_{i\in U_1} z_i + o_p(1).
  \end{equation*}
  \item[\textnormal{(P5)}] The matrix $\bM_q = \sum_{i\in U_1}\pi_i^{(p)} q_i \bx_i\bx_i^\top$ is positive definite uniformly in $N$.
\end{itemize}
Then
\begin{equation}\label{eq:consist_Bcom}
\widehat{\bB}_{q;\text{com}} = \bB_q + o_p(1).
\end{equation}
\end{proposition}

\subsection{Proof of Proposition~\ref{prop:consistency-combined}}\label{app:prop-consistency}

\begin{proof}
Define the normalized quantities
\begin{equation*}
\hat{\bM}_{q,N} := \frac{1}{N_1}\sum_{i\in S_{np}} q_i \bx_i \bx_i^\top + \frac{1}{N_1}\sum_{i\in S_p} q_i \bx_i \bx_i^\top,\qquad
\hat{\btau}_{q,N} := \frac{1}{N_1}\sum_{i\in S_{np}} q_i \bx_i y_i + \frac{1}{N_1}\sum_{i\in S_p} q_i \bx_i y_i,
\end{equation*}
based on the sample $S=S_{np}\cup S_p$, and their population counterparts
\begin{equation*}
\bM_{q,N} := \frac{1}{N_1}\sum_{i\in U_1}\pi_i^{(p)} q_i \bx_i\bx_i^\top,\qquad
\btau_{q,N} := \frac{1}{N_1}\sum_{i\in U_1}\pi_i^{(p)} q_i \bx_i y_i.
\end{equation*}
By \textnormal{(P4)} with $n_p\to\infty$, together with \textnormal{(P2)} and \textnormal{(P3)},
\begin{equation*}
\frac{1}{N_1}\sum_{i\in S_p} q_i \bx_i \bx_i^\top = \bM_{q,N} + o_p(1),\qquad
\frac{1}{N_1}\sum_{i\in S_p} q_i \bx_i y_i = \btau_{q,N} + o_p(1).
\end{equation*}
Since $n_{np}\to\infty$, a weak law of large numbers (LLN) together with \textnormal{(P1)} yields
\begin{equation*}
\frac{1}{n_{np}} \sum_{i\in S_{np}} q_i \bx_i \bx_i^\top \xrightarrow{p} \bA,\qquad
\frac{1}{n_{np}} \sum_{i\in S_{np}} q_i \bx_i y_i \xrightarrow{p} \bA \bbeta,
\end{equation*}
where $\bA := E_\xi(q_i \bx_i \bx_i^\top \mid i\in U_1)$, which is positive definite by (P5) together with (P3).
Introducing $\rho_N := n_{np}/N_1 \ge 0$, we can write
\begin{equation*}
\frac{1}{N_1}\sum_{i\in S_{np}} q_i \bx_i \bx_i^\top = \rho_N \bA + o_p(1),\qquad
\frac{1}{N_1}\sum_{i\in S_{np}} q_i \bx_i y_i = \rho_N \bA \bbeta + o_p(1).
\end{equation*}
Therefore, collecting the convergence results, we have
\begin{equation}\label{eq:M_tau_conv}
\hat{\bM}_{q,N} = \bM_{q,N} + \rho_N \bA + o_p(1),\qquad
\hat{\btau}_{q,N} = \btau_{q,N} + \rho_N \bA \bbeta + o_p(1).
\end{equation}
Write $y_i=\bx_i^\top\bbeta+e_i$ with $E_\xi(e_i\mid\bx_i)=0$. Then, by \textnormal{(P2)} and a weak law applied to $N_1^{-1}\sum_{i\in U_1}\pi_i^{(p)}q_i\bx_i e_i$, we obtain
\begin{equation}\label{eq:tau_conv}
\btau_{q,N} = \bM_{q,N}\,\bbeta + o_p(1).
\end{equation}

Since $\bM_q$ is positive definite by \textnormal{(P5)}, so is $\bM_{q,N}$, hence $\bM_{q,N}+\rho_N\bA$ for every $\rho_N\ge 0$. For any subsequence with $\rho_N\to\rho^\star\in[0,\infty)$, using \eqref{eq:M_tau_conv} and \eqref{eq:tau_conv},
\begin{equation*}
\hat{\bM}_{q,N} \xrightarrow{p} \bM_{q,N} + \rho^\star \bA,\qquad
\hat{\btau}_{q,N} \xrightarrow{p} (\bM_{q,N} + \rho^\star \bA)\bbeta,
\end{equation*}
and thus
\begin{equation*}
\widehat{\bB}_{q;\text{com}} = \hat{\bM}_{q,N}^{-1}\hat{\btau}_{q,N} \xrightarrow{p}
(\bM_{q,N} + \rho^\star \bA)^{-1}(\bM_{q,N} + \rho^\star \bA)\bbeta = \bbeta.
\end{equation*}
Since every subsequence has the same probability limit, we have $\widehat{\bB}_{q;\text{com}} \xrightarrow{p} \bbeta$. Moreover, by \textnormal{(P1)}, \textnormal{(P2)}, and \textnormal{(P5)}, $\bB_q \xrightarrow{p} \bbeta$, so \eqref{eq:consist_Bcom} holds.
\end{proof}

\subsection{Proof of Lemma~\ref{lem:unifV}}\label{app:lem-unifV}

\begin{proof}
Write $m_i=\bx_i^\top\bbeta$ and $\hat m_i=\bx_i^\top\hat{\bbeta}_{np}$. By assumptions (L1) and (L3)--(L4), the pilot estimator satisfies $\hat{\bbeta}_{np}\xrightarrow{p}\bbeta$. Therefore,
\begin{equation}\label{conv}
\sup_{i\in U_1}\big|\,\hat m_i-m_i\,\big|
\;\le\;\|\hat{\bbeta}_{np}-\bbeta\|\,\sup_{i\in U_1}\|\bx_i\|
\;\xrightarrow{p}\;0 ,
\end{equation}
using (L1). Fix any $\varepsilon\in(0,c)$. Since $\inf_{i\in U_1} m_i\ge c$ by (L2) and \eqref{conv} holds,
\begin{equation*}
P_\xi\!\Big(\min_{i\in U_1}\hat m_i \;\ge\; \min_{i\in U_1}m_i - \sup_{i\in U_1}|\hat m_i-m_i| \;\ge\; c-\varepsilon\Big)\to1,
\end{equation*}
hence, taking $\varepsilon=c/2$,
\begin{equation*}
P_\xi\!\Big(\min_{i\in U_1}\hat m_i \ge c/2\Big)\to1.
\end{equation*}

Under \eqref{eq:pow_mod} and (L3)--(L4), a standard M-estimation argument for the log-log variance fit yields $(\hat\gamma,\hat\sigma^2)\xrightarrow{p}(\gamma,\sigma^2)$. Since eventually $\hat m_i\in[c/2,M]$ with deterministic $M>C\|\bbeta\|$ and $\hat\gamma\in[\gamma/2,2\gamma]$, the map $(m,\gamma)\mapsto m^\gamma$ is uniformly continuous on $[c/2,M]\times[\gamma/2,2\gamma]$, hence
\begin{equation*}
\sup_{i\in U_1}\big|\,\hat m_i^{\hat\gamma}-m_i^{\gamma}\,\big|\xrightarrow{p}0.
\end{equation*}
Finally,
\begin{equation*}
\sup_{i\in U_1}\big|\,\hat\sigma^2\,\hat m_i^{\hat\gamma}-\sigma^2 m_i^{\gamma}\,\big|
\;\le\;
|\hat\sigma^2-\sigma^2|\,\sup_{i\in U_1}\hat m_i^{\hat\gamma}
+\sigma^2\,\sup_{i\in U_1}\big|\hat m_i^{\hat\gamma}-m_i^{\gamma}\big|
\;\xrightarrow{p}\;0,
\end{equation*}
because $\sup_{i\in U_1}\hat m_i^{\hat\gamma}=O_p(1)$ due to bounded covariates and $\hat\gamma=O_p(1)$. This proves the claim.
\end{proof}

\subsection{Proof of Theorem~\ref{th:1}}\label{app:thm-opt}

\begin{proof}
By Lemma~\ref{lem:unifV} and the continuous mapping theorem,
\begin{equation}\label{eq:sqrtVconv}
\sup_{i\in U_1}\Big|\{\hat V_\xi(y_i\mid\bx_i)\}^{1/2}-\{V_\xi(y_i\mid\bx_i)\}^{1/2}\Big|\xrightarrow{p}0.
\end{equation}
Let
\begin{equation*}
S_V=\sum_{j\in U_1}\{V_\xi(y_j\mid\bx_j)\}^{1/2},
\qquad
\hat S_V=\sum_{j\in U_1}\{\hat V_\xi(y_j\mid\bx_j)\}^{1/2}.
\end{equation*}
From \eqref{eq:sqrtVconv},
\begin{equation*}
\frac{\hat S_V-S_V}{N_1}
\;\le\;
\sup_{j\in U_1}\Big|\{\hat V_\xi(y_j\mid\bx_j)\}^{1/2}-\{V_\xi(y_j\mid\bx_j)\}^{1/2}\Big|
\;\xrightarrow{p}\;0.
\end{equation*}
By (L1)--(L2) and \eqref{eq:pow_mod}, the terms
$\{V_\xi(y_i\mid\bx_i)\}^{1/2}$ are uniformly bounded and bounded away from zero on $U_1$, so $S_V/N_1$ and $\hat S_V/N_1$ have finite nonzero probability limits. Therefore,
\begin{equation}\label{eq:unifpiconv}
\max_{i\in U_1}\big|\hat{\pi}_i^{(p)}-\pi_i^{(p)}\big|
=
n_p\,\max_{i\in U_1}
\left|
\frac{\{\hat V_\xi(y_i\mid\bx_i)\}^{1/2}}{\hat S_V}-\frac{\{V_\xi(y_i\mid\bx_i)\}^{1/2}}{S_V}
\right|
\;\xrightarrow{p}\;0,
\end{equation}
so $\hat{\pi}_i^{(p)}\xrightarrow{p}\pi_i^{(p)}$ uniformly in $i\in U_1$.
The Isaki--Fuller probabilities \eqref{eq:opt_pi} minimize the anticipated variance of $\hat Y_{\text{sepDI}}$ over designs on $U_1$. By \eqref{eq:unifpiconv}, the estimated design converges in probability to that optimal design; under (T1)--(T2), the anticipated variance functional is continuous in the design probabilities (they are uniformly bounded away from zero). Hence, the anticipated variance under \eqref{eq:est_opt_pi} converges to the minimum. Because (T2) keeps the sampling fraction nondegenerate, the design does not become sparse; therefore \eqref{eq:est_opt_pi} is asymptotically optimal for $\hat Y_{\text{sepDI}}$.
\end{proof}

\bibliographystyle{apalike}  
\bibliography{literature2}

\end{document}